\documentclass[11pt]{article}
\usepackage{amsmath}
\usepackage{amssymb}
\usepackage{amsthm}
\usepackage{pdfsync}
\usepackage{graphicx}
\usepackage{caption}
\usepackage{comment}
\usepackage[left=0.9in, right=0.9in, top=0.65in, bottom=0.8in]{geometry}
\usepackage[fit, breakall]{truncate}
\usepackage{fancyhdr}

\allowdisplaybreaks[3]

 \newtheorem{thm}{Theorem}[section]

 \newtheorem{prop}[thm]{Proposition}
 \theoremstyle{definition}
 
 \theoremstyle{remark}
 
 \numberwithin{equation}{section}
 \DeclareMathOperator{\RE}{Re}
 \DeclareMathOperator{\IM}{Im}

 \newcommand{\h}{\mathcal{H}}

 \newcommand{\abs}[1]{\left\vert#1\right\vert}

\title{Classical Elastic Two-Particle Collision Energy Conservation using Edward Nelson's Energy, Double Diffusion and Special Relativity}

\author{J\lowercase{ohan} G.B. B\lowercase{eumee}*, H\lowercase{erschel} R\lowercase{abitz} \\ \\
\textit{\textsl{P\lowercase{rinceton} U\lowercase{niversity}}}
\\
\textit{\textsl{P\lowercase{rinceton}, N\lowercase{ew} J\lowercase{ersey}} }
\\
\textit{\textsl{USA} }}

\date{}

\begin{document}
\maketitle

\fancyhead[RE, RO]{\tiny \textbf{Classical Elastic Two-Particle Collision Energy ... \thepage}}

\providecommand{\keywords}[1]
{
  \small	
  \textbf{\textit{Keywords---}} #1
}
\begin{abstract}
\smallskip
The present paper shows that Edward Nelson’s stochastic mechanics approach for
quantum mechanics can be derived from two classical elastically colliding particles with
masses M and m satisfying a collision momentum-preserving equation. The properties
of the classical elastic momentum collision expression determine the full Edward Nelson
energy collision energy for both particles. The classical Edward Nelson total energy ex-
pression does not require a statistical expectation since no process was defined for the
energy and it models the main and incident particle velocities perfectly. The classical
Edward Nelson energy expression can easily derive quantum mechanics as a special case
and show classical physics results for mechanics, collision theory, cosmology and relativity. Quantum mechanics can be obtained by modelling the incident particle as a (non-random) potential using stochastic processes modelling the forward (post-collision) and backward (pre-collision) velocities of the main particle. This presents the Schr\"odinger equation exactly the way that Nelson proposed in 1966 except for the diffusion constant. In this case the average energy is conserved in time and the forward (post-collision) and backward (pre-collision) velocities of the system are related using statistical methods. If the incident particle does not have a potential it would be convenient to assume another stochastic process for the incident particle again requiring that the average energy is preserved. In this case the additional constraints for the movement of the incident particle leads to another Schr\"odinger equation. Two Schr\"odinger equations would reserve the energy conditions and two potentials may be introduced as well since the behaviour of independent particles is independent (correlation equals zero). Finally, under suitable conditions it will be shown that the colliding particles satisfy Minkowski’s metric in special relativity. This last example shows how gravity can be quantized using details of this energy expression.
\end{abstract}


\section*{Introduction}
\smallskip
To find a classical environment for the physics of Edward Nelson's stochastic mechanics consider the classical double particle elastic collision between a main particle of mass $M$ and the incident particle $m$. Both particles are perfectly spherical objects so the perfect elastic classical collision only deals with momentum and ignores angular momentum or spin. During the collision the projected velocity terms of the main velocity $v_2, w_2$ and incident velocities $v_1, w_1$ collide along the (one-dimensional) collision axis $\phi$ representing the center of mass connection leaving the remaining momentum and energy undisturbed. In section 1 the first result uses the projection $P(\phi)$ on the collision axis $\phi$ to show the classical momentum conservation for the two classical colliding particles. The paper then shows in Appendices A, B that the momentum conservation equation implies that the total energy at the time of collision equals the Edward Nelson stochastic mechanics energy. In other words the momentum preserving equation is instrumental in representing the main and incident particle energy into the Edward Nelson stochastic mechanics energy without referring to quantum mechanics.

Section 2 shows that the original Nelson quantum mechanical Schr\"odinger equation can be obtained assuming that the second particle can be approximated as a (non-random) potential while the main particle motion is represented by a single diffusion process. The energy expression has the familiar form of a forward and backward (drift) velocities for the main particle and a potential term. The diffusion parameters for the Schr\"{o}dinger equations has the same diffusion parameter $\eta$ but the variance of the particles is different. The elements in the energy sums are weighted by the particle masses $M$, $m$ and the proof of this conserved energy is demonstrated in Appendix C. Nelson showed in 1966 that the diffusion constant can be adapted to derive Schr\"odinger's equation by choosing $\eta = \hbar$. This demonstrated that quantum mechanics can be derived from the forward and backward drifts of very simple stochastic processes with the drifts showing the correct pre- and post collision velocities.

As section 3 demonstrates, if the incident particle is represented with a stochastic process instead of a (non-random) potential  the particle can behave very similar to the main particle. A straightforward solution is to choose independent diffusion processes for both the main and incident particle and then investigate using the forward and backward Fokker-Planck equations as well as the diffusion rate of these equations. So there are two processes for the main and incident particles and the probability density to constrain the energy means solving a double set of Schr\"{o}dinger equations. The diffusion parameters for the Schr\"{o}dinger equations for both particles have the same diffusion parameter $\eta$ but the variance of the particles is different. This paper shows that it is possible to combine many types of Schr\"odinger equations for a solution of the combined equation.

Section 4 shows that a main particle moving through a medium of incident particles is required to be correlated to the incident particle. First assume that the correlation equals zero, then the momentum equation requires that the outgoing particle $M\abs{v}^2\approx m\abs{w}^2$ so that the main particle adopts the environment energy. If the main particle moves too slow it will accelerate and otherwise it will slow down. Now assume that the particle has a constant speed and an undetermined correlation. The correlation between the main and incident particle is complicated for masses where $m \approx M$ but a very clear approximation can be derived when the incident particle is small so that $m << M$. In that case the mass ratio becomes $\gamma_m^2=m/M \ll 1$ and the correlation between the motion of the main particle with the incident particle is almost linear.

Section 5 shows that the conserved classical momentum equation can be solved employing the main particle Eigenvalues / Eigenvectors of the main scattering matrix showing both the average motion of both particles and a set of vectors that are invariant on the average velocity. For the Eigenvalue equal 1 the supporting Eigenvector equals $(a,a)$ where $(1+\gamma^2)a = Mv_1+mw_1 = Mv_2+mw_2$ hence this Eigenvector occurs in the ingoing and outgoing velocities. The second Eigenvalue equals -1 and the supporting  Eigenvector becomes $(-\gamma_m^2g,g)$ where $g \sim v_1-w_1$. Then if the incoming vector becomes proportional to the sum of $(a,a)$ and $(-\gamma_m^2g,g)$, the outgoing set of velocities becomes the sum of $(a,a)$ and $(\gamma_m^2g^\bot,-g^\bot)$ where $g^\bot \sim (w_2-v_2)$. The paper shows that $\abs{g}^2=\abs{g^\bot}^2$ and this equality is used to determine the Minkowski equality.

As mentioned above, the original Nelson result orignated a large activity around stochastic mechanics Nelson~\cite{ENELSON3}, ~\cite{ENELSON2},~\cite{ENELSON1}, Carlen~\cite{ECARL1}, ~\cite{ECARL2}, ~\cite{ECARL3}, Guerra~\cite{GUERRA1}, ~\cite{GUERRA2}, Vaidman~\cite{Vaidman1}, Smolin~\cite{SMOLIN1}, Albert~\cite{ALBERT1}. Carlen in particular emphasized that this construction for the examples would almost always have weak solutions as long as the potentials have the Kato-Rellich property. There was an effort to derive quantum mechanics from first principles Albeverio~\cite{ALBEVERIO1}, ~\cite{ALBEVERIO2}, ~\cite{ALBEVERIO3}, Morato~\cite{MORATO2}, Blanchard~\cite{BLANCHARD1}, Chen~\cite{CHEN1}, Davidson~\cite{DAVIDSON1}, de la Pena~\cite{DELAPENA1}, Dohrn~\cite{DOHRN1}, Zambrini~\cite{ZAMBRINI1}, ~\cite{ZAMBRINI2}, Mansi~\cite{MANSI1} and Baccigaluppi~\cite{BACCIAGALUPPI1}. Earlier efforts started with Fenyes~\cite{FENYES1} and stochastic mechanics was criticised by Wallstrom~\cite{WALLSTROM1}. The use of stochastic processes in physics are extensively referenced by Einstein~\cite{EINST1}, Avrahan~\cite{AVRAHAN1}, Beumee~\cite{BEUMEE1}, Carmona~\cite{CARMONA1}, Feller~\cite{FELLER1}, Feynman~\cite{FEYNMAN1}, Goldstein~\cite{GOLDSTEIN1}, van Kampen~\cite{KAMPEN1}, Lanczos~\cite{LANCZOS1}, Posilicano~\cite{POSILICANO1} and Rogers~\cite{ROGERS1}, Shreve~\cite{SHREVE1}, Thide~\cite{THIDE1} and Vastola~\cite{VASTOLA1} among many others.

Further work considered particle spin Fritsche~\cite{FRITSCHE1}, relativity Garbaczewski~\cite{GARBA1},~\cite{GARBA2}, Kuipers~\cite{KUIPERS1}, Gamba~\cite{GAMBA1}, Marra~\cite{MARRA1}, Serva~\cite{SERVA1}, Minkowski~\cite{GUERRA3}, stochastic variations Morato~\cite{MORATO1}, Oriols~\cite{ORIOLS1}, Bhattacharya~\cite{BHATTACHARYA1} and chaos analysis by Nottale~\cite{NOTTALE1}. In addition Ohsumi~\cite{OHSUMI1} investigated stochastic system theory, Pavon~\cite{PAVON1} investigated the Feynman path integral and stochastic entanglement was investigated by Penrose~\cite{PENROSE1}. On cosmological matters stochastic mechanics investigated dark matter, see Paredes~\cite{PAREDES1}, Cresson~\cite{CRESSON1}, ~\cite{CRESSON2}, Chavanis~\cite{CHAVANIS1}, Chamaraux~\cite{CHAMARAUX1} and there were studies for Schr\"{o}dinger-Newton Robertshaw~\cite{ROBERTSHAW1}, Tod\cite{TOD1} and Harrison~\cite{HARRISON1}. There was also work on biology Roy~\cite{ROY1}, Caticha~\cite{CATICHA2}, quantum gravity Calcagni~\cite{CALCAGNI11}, Colella~\cite{COLELLA1}, field theory by Dijkgraaf~\cite{DIJKGRAAF1} and philosophy Puthoff~\cite{PUTHOFF1}. A recent thesis focussed on stochastic mechanics Derakhshani~\cite{DERAKHSHANI1}, ~\cite{DERAKHSHANI2}, ~\cite{DERAKHSHANI3} and research was done on the Burgers equation by Prodanov~\cite{PRODANOV1}.

\medskip

\section{Energy/Momentum Conservation}
This section shows the momentum distribution for classical elastic collision Hamiltonian assuming the main particle of mass $M$ and the incident particle of mass $m$. For classical particles the classical momenta is required to match the momenta $\mathcal{P}_1, \mathcal{P}_2$ and the energy terms $\h_{1}, \h_{2}$. This means that the main particle $v_2,v_1\in\Re^3$ and the incident particle $w_1,w_2\in\Re^3$ satisfy
\begin{subequations}
\label{l:ENERG1P}
\begin{align}
\label{l:ENERG1a}
&\mathcal{P}_1=p_1+q_1=Mv_1+mw_1,
\\
\label{l:ENERG1b}
&\mathcal{P}_2=p_2+q_2=Mv_2+mw_2,
\\
\label{l:ENERG1c}
&\h_{1}=\frac{1}{2}\left(M\abs{v_2}^2+m\abs{w_2}^2\right),
\\
\label{l:ENERG1d}
&\h_{2}=\frac{1}{2}\left(M\abs{v_1}^2+m\abs{w_1}^2\right),
\end{align}
\end{subequations}
and the classical elastic collision momentum and energy conservation requires that $\mathcal{P}_2=\mathcal{P}_1$ and $\h_{2}=\h_{1}$.

Imagine, that the collision between the perfectly spherically particles happens elastically where the main particle with mass $M$ and the incident particle with mass $m$ collide with hard surfaces. This model assumes that angular momentum or spin is not incorporated. The elastic collisions of exchanging the momentum and energy only occurs in the direction $\phi \in \Re^3$ stretching from the center of the main particle to the center of the incident particle. To conveniently arrange for motion along a random unit-size vector $\phi$ consider the projection $P(\phi)=\phi\phi^T$ which maps all vectors onto the vector $\phi$. Clearly $P(\phi)$ is a projection because $P(\phi)P(\phi)=(\phi\phi^T)(\phi\phi^T)=\phi(\phi^T\phi)\phi^T=P(\phi)$ since $\phi$ is a unit vector. Also $P(\phi)(I-P(\phi))=(P(\phi)-P(\phi)) = 0$ so the operator $(I-P(\phi))$ is always orthogonal to $P(\phi)$ hence $(I-P(\phi))v\perp P(\phi)w$ for all three dimensional vectors $v, w \in \Re^3$.

To implement the conservation law along the $\phi$ axis we decompose the $v_1,v_2,w_1$ and $w_2$  vectors as follows
\begin{align*}
v_1&= P(\phi)v_1 + (I-P(\phi))v_1,
\\
v_2&= P(\phi)v_2 + (I-P(\phi))v_2,
\\
w_1&= P(\phi)w_1 + (I-P(\phi))w_1,
\\
w_2&= P(\phi)w_2 + (I-P(\phi))w_2,
\end{align*}
and for all classical momentum and energy exchanged orthogonal to $\phi$ clearly
\begin{align*}
(I-P(\phi))v_2 & = (I-P(\phi))v_1,
\\
(I-P(\phi))w_2 & = (I-P(\phi))w_1.
\end{align*}

Along $\phi$ the conservation of momentum $\mathcal{P}_1, \mathcal{P}_2$ and energy $\h_1, \h_2$ leads to
\begin{align*}
MP(\phi)v_1 + mP(\phi)w_1 &= MP(\phi)v_2 + mP(\phi)w_2,
\\
Mv_1^TP(\phi)v_1 + mw_1^TP(\phi)w_1 &= Mv_2^TP(\phi)v_2 + mw_2^TP(\phi)w_2.
\end{align*}
The last equation leads to
\begin{align*}
M\phi^Tv_1 + m\phi^Tw_1 &= M\phi^Tv_2 + m\phi^Tw_2,
\\
M(\phi^Tv_1)^2 + m(\phi^Tw_1)^2 &= M(\phi^Tv_2)^2 + m(\phi^Tw_2)^2,
\end{align*}
making use of the $P(\phi)$ projection property again by removing $\phi$ from the first equation while the second equation was created by substituting $P(\phi)=\phi^T\phi$. Now the conservation problem is one-dimensional and Appendix A shows that the solution equals
\begin{align}
\label{App14}
\begin{pmatrix}
\phi^Tv_2 \\
\phi^Tw_2
\end{pmatrix}
&= \begin{pmatrix}
\cos(\theta) & \gamma_m\sin(\theta) \\
\frac{\sin(\theta)}{\gamma_m} & -\cos(\theta)
\end{pmatrix}
\begin{pmatrix}
\phi^Tv_1\\
\phi^Tw_1
\end{pmatrix},
\end{align}
where $\gamma_m^2 = m/M$, $\sin\left(\theta\right)
=2\gamma_m/\left(1+\gamma_m^2\right)$,
$\cos\left(\theta\right)=\left(1-\gamma_m^2\right)/\left(1+\gamma_m^2\right)$ and $\phi$ is the main to incident particle unit vector. The matrix here is referred as the collision matrix. Notice that this equation does not reflect the angular momentum assuming that this motion is independent of the particle momenta. From that point of view this equation is an approximation. The results are summarized in the following Theorem.

\begin{thm}
\label{l:THEOREM3}
In three dimensions for a classical elastic collision where $\mathcal{P}_1=\mathcal{P}_2$ (classical momentum) and $\h_1=\h_2$ it follows that the output for the main and incident particles $v_2,w_2$ must be related to the input variables $v_1$ and $w_1$. In this case
\begin{align}
\label{l:EQMOT9}
\begin{split}
\begin{pmatrix}
v_2 \\
w_2
\end{pmatrix}
&= \begin{pmatrix}
\cos(\theta) & \gamma_m\sin(\theta) \\
\frac{\sin(\theta)}{\gamma_m} & -\cos(\theta)
\end{pmatrix}
\begin{pmatrix}
v_1\\
w_1
\end{pmatrix}
+
\begin{pmatrix}
\gamma_m \Phi
\\
-\frac{1}{\gamma_m} \Phi
\end{pmatrix}
\\
&=\begin{pmatrix}
I-\gamma_m\sin(\theta)P(\phi) & \gamma_m\sin(\theta)P(\phi) \\
\frac{\sin(\theta)}{\gamma_m}P(\phi) & I-\frac{\sin(\theta)}{\gamma_m}P(\phi)
\end{pmatrix}
\begin{pmatrix}
v_1
\\
w_1
\end{pmatrix},
\end{split}
\end{align}
where $I$ is the unit matrix and $\Phi = \sin(\theta) (I-P(\phi))(v_1-w_1)$ is a 3-dimensional vector so that $\Phi^T(v_2-v_1)=0$. Here $P(\phi)$ is the random projection in the direction of $\phi$ with $\gamma_m^2 = m/M$, $\sin\left(\theta\right) =2\gamma_m/\left(1+\gamma_m^2\right)$ and $\cos\left(\theta\right)=\left(1-\gamma_m^2\right)/\left(1+\gamma_m^2\right)$.
\end{thm}
\begin{proof}
For a detailed proof see Appendix A.
\end{proof}

Subtracting the terms $v_2$ and $\gamma_m^2w_2$ shows that the classical momentum conservation $\gamma_m^2(w_2-w_1) = -(v_2 - v_1)$ while the classical energy is conserved due to the quadratic property of $\Phi$, see \eqref{l:DRFT24}. Equation \eqref{l:EQMOT9} has some alternative representations as follows.
\begin{prop}
\label{l:THEOREM4}
Using Theorem \ref{l:THEOREM3} the elastic collision solution to the four equations in equation \eqref{l:ENERG1P} can also be represented as
\begin{align}
\label{l:EQMOT1P}
\begin{split}
\begin{pmatrix}
v_2 \\
w_2
\end{pmatrix}
&=
\begin{pmatrix}
v_1 \\
w_1
\end{pmatrix}
+
\begin{pmatrix}
\gamma_m \sin(\theta)P(\phi)(w_1-v_1) \\
- \frac{1}{\gamma_m}\sin(\theta)P(\phi)(w_1-v_1)
\end{pmatrix},
\end{split}
\end{align}
with $\phi$ showing the (random, unit-size) connection between the particles and $\gamma_m^2=m/M$,
$\sin\left(\theta\right) = 2\gamma_m/\left(1+\gamma_m^2\right)$,
$\cos\left(\theta\right)=\left(1-\gamma_m^2\right)/\left(1+\gamma_m^2\right)$.
As a result
\begin{align}
\label{l:EQMOT3}
\begin{split}
\begin{pmatrix}
E[v_2 \vert x] \\
E[w_2\vert x]
\end{pmatrix}
\approx
\begin{pmatrix}
E[v_1\vert x] \\
E[w_1\vert x]
\end{pmatrix}
+
\begin{pmatrix}
\gamma_m \sin\left(\theta\right) (E[w_1\vert x]-E[v_1\vert x]) \\
- \frac{1}{\gamma_m}\sin\left(\theta\right)(E[w_1\vert x]-E[v_1\vert x])
\end{pmatrix}
\approx
\begin{pmatrix}
E[v_1\vert x] \\
E[w_1\vert x]
\end{pmatrix},
\end{split}
\end{align}
if $E[P(\phi)]\rightarrow I$ and $E[w_1\vert x] \approx E[v_1\vert x]$ for all $x$.
\begin{proof}
For a detailed proof see Appendix B. The first equation is just a rewrite of the second term in \eqref{l:EQMOT9} while the second term is created from the first term using the definition of $\Phi$.
\end{proof}
\end{prop}
If the average motion of the main particle $v_2,v_1$ is not as large as $w_2,w_1$ the main particle $M$ speeds up and the incident particle slows down otherwise the reverse happens.

\section{Classical Energy/Momentum with Nelson's Measure}
The simple equation \eqref{l:EQMOT9} can be solved easily once $P(\phi)$ is established, but it is possible to find an alternative representation. The following Theorem is the main result of this paper describing how the classical energy and momentum conservation for the classical energy of the main and incident particles can be expressed exclusively in terms of $v_2$ and $v_1$. Notice that this classical energy constraint ignores any terms depending on the random projection $P(\phi)$ and therefore ignores the random factor $\Phi$ in equation \eqref{l:EQMOT9}.

\begin{thm}
\label{l:THEOREM5}
Let the elastic collision momentum and energy defined in equations \eqref{l:ENERG1P} hold so that the pre-collision momentum equals the post-collision momentum $\mathcal{P}_2=\mathcal{P}_1$ and the pre-collision energy equals the post-collision energy $\h_1=\h_2$. It then follows that the classical conserved energy of both the main particle and the incident particle equals the sum of both particles classical energy expressed as
\begin{align}
\begin{split}
\label{l:DRFT4}
\frac{2\h_2}{M} & = \frac{2\h_1}{M}
=\abs{\frac{v_2+v_1}{2}}^2+
\abs{\frac{v_2-v_1}{2}}^2 + \gamma_m^2\left(\abs{\frac{w_2+w_1}{2}}^2+
\abs{\frac{w_2-w_1}{2}}^2\right),
\end{split}
\end{align}
where $\gamma_m^2 = m/M$. Here the main and incident energies are expressed as a sum of squares of symmetric and anti-symmetric (osmotic) summations.

\begin{proof}
See Appendix C for a detailed proof.
\end{proof}
\end{thm}
As mention before this equation can be altered with solutions of the classical conserved momentum in \eqref{l:EQMOT9}. An amount of osmotic energy contribution $\alpha \abs{v_2-v_1}^2$ in \eqref{l:DRFT4} can be replaced by $\alpha\gamma_m^2\abs{w_2-w_1}^2$ and imply the same energy constraint except for parameters in the osmotic energy terms. Hence the solution obtained is not unique as the energy and momenta can be exchanged and added accordingly.

It is straightforward to show that equation \eqref{l:DRFT4} reduces to the Nelson stochastic mechanics potential by assuming that the second terms in $w_1, w_2$ can be replaced by a non-random function so that the energy can be represented by an expectation. Assume the potential $V_{incident}=V_{incident}(x,t)$ equal to
\begin{align}
\label{APP32}
\begin{split}
V_{incident}=V_{incident}(x,t) & \approx \frac{M\gamma_m^2}{2}E\bigg[\abs{\frac{w_2+w_1}{2}}^2 +\abs{\frac{w_2-w_1}{2}}^2\bigg|x(t)=x\bigg]
\\
& = mE\bigg[\abs{\frac{w_2+w_1}{2}}^2
+\abs{\frac{w_2-w_1}{2}}^2\bigg|x(t)=x\bigg],
\end{split}
\end{align}
then the incident particle energy is a non-random expectation. This is an approximation as the expectation does not reflect the forward or backward drift and ignores any drift terms. If the mass ratio is very small so that $m\ll M$ then the osmotic energy term and the expectation becomes very small so that the remaining energy term looks more non-random. This equation is therefore an approximation of the Hamiltonian $\h_2, \h_2$ so that
\begin{align}
\label{l:DRFT5}
\begin{split}
\widetilde{\h}_2 & = \widetilde{\h}_1
= \frac{M}{2}\left(\abs{\frac{v_2+v_1}{2}}^2+
\abs{\frac{v_2-v_1}{2}}^2\right) + V_{incident}.
\end{split}
\end{align}

For a time-independent solution to \eqref{l:DRFT5} assume that $v_1, v_2$ are represented by an independent process and assume that the incoming classical energy equals the outgoing classical energy so that $\widetilde{\h}_1=\widetilde{\h}_2$. Then the expectation over the energy needs to be determined over a set of diffusion processes so that
\begin{align}
\label{l:DRFT6}
\begin{split}
E[\widetilde{\h}_2] = E[\widetilde{\h}_1]=
&E\left[\frac{M}{2}\abs{\frac{v_2+v_1}{2}}^2+
\abs{\frac{v_2-v_1}{2}}^2 + V_{incident}\right]
\\
=&\frac{M}{2}E\left[\abs{\frac{v_2+v_1}{2}}^2+
\abs{\frac{v_2-v_1}{2}}^2\right] + \int \rho(x,t)V_{incident}(x,t)dx,
\end{split}
\end{align}
then the average classical energy over this expectation must be $\widetilde{\h}_1=\widetilde{\h}_2$. Following Nelson, the easiest way of solving this problem is to demand that $x(t)$ satisfies a diffusion process and calculate the average velocities $v_2 \approx x(t+\tau_2) - x(t), v_1 \approx x(t) - x(t-\tau_1)$ on the energy constraint \eqref{l:DRFT6}. The following Theorem calculates the energy constraint under these conditions.
\begin{thm}
\label{l:THEOREM6}
(Edward Nelson) Let the three-dimensional position density $\rho=\rho(x,t)$ of the main particle $x(t)$ depend on the forward drift $b^+=b^+(x,t)$ and backward drift $b^-=b^-(x,t)$ satisfying
\begin{subequations}
\begin{align}
\label{l:DRFT2aa}
v_2\tau_2 = x(t+\tau_2)-x(t)&=b^+\tau_2+\sigma d^+z(t),
\\
\label{l:DRFT2ab}
\begin{split}
v_1\tau_1 = x(t)-x(t-\tau_1) & = b^-\tau_1+\sigma d^-z(t),
\\
b^-&=b^+- \sigma^2\nabla_x\log{\rho},
\end{split}
\end{align}
where $d^+z(t)$ and $d^-z$(t) are independent Gaussian increments and $\sigma$ is the main particle position variance. The probability density is defined as the solution of
\begin{align}
\label{l:DRFT3aa}
\frac{\partial \rho}{\partial t}&= - \nabla_x^T\left(b^+\rho\right)
+\frac{\sigma^2}{2}\Delta\rho,
\\
\label{l:DRFT3ab}
\frac{\partial \rho}{\partial t}&= -\nabla_x^T\left(b^-\rho\right)
- \frac{\sigma^2}{2}\Delta\rho,
\end{align}
so that
\begin{align}
\label{l:DRFT3ac}
\frac{\partial \rho}{\partial t}= -
\nabla_x^T\left(\left(\frac{b^++b^-}{2}\right)\rho\right),
\end{align}
\end{subequations}
where $\nabla_x^T=\left(\frac{\partial}{\partial x_1},...,\frac{\partial}{\partial x_3}\right)$ and $\Delta=\sum_{i} \frac{\partial^2}{\partial x_i^2}$. Then using the potential $V_{incident}=V_{incident}(x)$ the expectation of the pre-collision energy $E[\widetilde{\h}_1]$ equals the post-collision average energy $E[\widetilde{\h}_2]$ so that
\begin{align}
\label{l:HAMILT1}
\begin{split}
E\left[\widetilde{\h}_1\right] & = E\left[\widetilde{\h}_2\right] =
\\
&\frac{M}{2}\left(E\left[\abs{\frac{b^++b^-}{2}}^2\right]
+\frac{\eta^2}{4M^2} E\left[\abs{\frac{1}{\rho}\nabla_x
\rho}^2\right]\right)
+E[V_{incident}(x)] +
\frac{3\eta}{\overline{\tau}},
\end{split}
\end{align}
where
\begin{align*}
\eta & = M\sigma^2,
\\
E[V_{incident}] & = \int \rho V_{incident}(x)dx.
\end{align*}
Here $E[.]$ shows the density of $\rho$ for state space x taking into account the forward $b^+$ and backward $b^-$ mean drifts and averages over time as $E[1/\tau]=2/\overline{\tau}$.
\end{thm}
\begin{proof}
The terms that require an explanation are the expectation of the kinetic energy, the osmotic component and the resident constant variance at the end of equation \eqref{l:HAMILT1}. Using the expressions for the main particle backward $b^-$ and forward drift $b^+$ defined in \eqref{l:DRFT2aa}, \eqref{l:DRFT2ab} then the energy \eqref{l:DRFT6} becomes
\begin{align*}
E\left[\abs{\frac{v_2+v_1}{2}}^2+ \abs{\frac{v_2-v_1}{2}}^2\right]
= &E\left[\abs{\frac{b^++b^-}{2}+ \frac{1}{2}\sigma\frac{\Delta^+z}{\tau_2}+\frac{1}{2}
\sigma\frac{\Delta^-z}{\tau_1} }^2\right]
\\
& +E\left[\abs{\frac{b^+-b^-}{2}
+\frac{1}{2}\sigma\frac{\Delta^+z}{\tau_2}-\frac{1}{2}
\sigma\frac{\Delta^-z}{\tau_1} }^2\right],
\end{align*}
which can be simplified to
\begin{align}
\label{l:HAMILT3}
\begin{split}
E&\left[\abs{\frac{b^++b^-}{2}}^2\right]
+E\left[\abs{\frac{b^+-b^-}{2}}^2\right]
+\frac{3\sigma^2}{2}E\left[\frac{1}{\tau}\right]
\\
=&E\left[\abs{\frac{b^++b^-}{2}}^2\right]
+\frac{\sigma^4}{4} E\left[\abs{\frac{1}{\rho}\nabla_x
\rho}^2\right]
+\frac{3\sigma^2}{\overline{\tau}},
\end{split}
\end{align}
assuming that $E[1/\tau] = 2 / \overline{\tau}$ and using the definition for $\sin(\theta)=2\gamma_m/(1+\gamma_m^2)$ provided in Theroem \ref{l:THEOREM3}.

Now we have
\begin{align*}
E[\widetilde{\h}_1] & = E[\widetilde{\h}_2] =
\frac{M}{2}E\left(\abs{\frac{v_2+v_1}{2}}^2+
\abs{\frac{v_2-v_1}{2}}^2\right)
\\
&=\frac{M}{2}\left(E\left[\abs{\frac{b^++b^-}{2}}^2\right]
+\frac{\sigma^4}{4}E\left[\abs{\frac{1}{\rho}\nabla_x
\rho}^2\right]\right)
+\frac{3M\sigma^2}{\overline{\tau}}
\\
&=\frac{M}{2}\left(E\left[\abs{\frac{b^++b^-}{2}}^2\right]
+\frac{\eta^2}{4M^2} E\left[\abs{\frac{1}{\rho}\nabla_x
\rho}^2\right]\right)
+\frac{3\eta}{\overline{\tau}},
\end{align*}
which reconciles the main expected particle average velocity and osmotic terms.

The integral of the potential term is more straightforward as it depends only on the position and time of the stochastic process $x(t)$ so
\begin{align*}
E\left[V_{incident}\right]=\int \rho V_{incident}(x,t)dx,
\end{align*}
and the Theorem is proven.
\end{proof}

Theorem \ref{l:THEOREM6} refers to any classical energy and momentum conserving set of particles so it can be applied to all Hamiltonian collision problems, diffusion problems, quantum mechanics, electrodynamic problems and gravity. There is no reference here to a stochastic interpretation of quantum mechanics though it is tempting to formulate one. Following Nelson it is clear that making equation \eqref{l:HAMILT1} time-independent refers to the Schr\"{o}dinger wave function if the variance equals $\eta = \hbar$ which is possible since
\begin{align}
\label{APP29}
\sigma^2 = \frac{\eta}{M} = \frac{\hbar}{M}.
\end{align}
However, the statement here is an approximation of the real energy constraint in \eqref{l:DRFT4} and it applies to classical energy proven by the momentum equation \eqref{l:EQMOT9}. That makes the result applicable to all mechanical applications not just quantum mechanics.

Equation \eqref{l:HAMILT1} refers to the constant energy term in the energy function $3\eta/\overline{\tau}$ which is proportional to the particle variance $\sigma^2$ and inversely proportional to the inter-particle average collision time $\overline{\tau}$. This constant contribution is based on assuming that the inter-particle collisions happen where $\overline{\tau}$ is the average time between collisions and assumes that $\tau$ satisfies a two-dimensional Gamma distribution with a mean of $\overline{\tau}$. This modelling or estimated average inter-particle $\overline{\tau}$ may change as a function of time so that it is possible that the collision energy $E\left[\widetilde{\h}_1\right] = E\left[\widetilde{\h}_2\right] \thicksim 3\eta/\overline{\tau}$ increases or decreases as a function of collision times.

The following Theorem introduces the phase function $S=S(x,t)$ and amplitude $R=R(x,t)$ to create a more elegant form of the classical energy condition \eqref{l:HAMILT1}.
\begin{thm}
\label{l:HAMILT12}
(Edward Nelson) Assume that the main particle $v_2, v_1$ behaves like a diffusion processes with a forward drift $b^+=b^+(x,t)$, a backward drift $b^-=b^-(x,t)$ and a variance of $\sigma$. Then \eqref{l:HAMILT1} becomes time-independent if there exists a wave density $\psi(x,t)$ so that
\begin{align}
\label{l:DRFT23}
\psi &=\psi(x,t)=
\exp\left[\frac{M\left(R + iS\right)}{\eta}\right] = \exp\left[\frac{\left(R + iS\right)}{\sigma^2}\right],
\end{align}
where
\begin{subequations}
\begin{align}
\label{l:APP13a}
2\nabla_x R & = \sigma^2\frac{\nabla_x
\rho}{\rho} =b^+-b^-,
\\
\label{l:APP13b}
2\nabla_x S & = \left(b^++b^-\right),
\end{align}
and where the main particle satisfies Schr\"{o}dinger's equation
\begin{align}
\label{l:SCHROD5}
\begin{split}
i\eta\frac{\partial\psi}{\partial t}&=-\frac{\eta^2}{2M}\Delta_x\psi + V(x)\psi,
\\
\rho&=\,|\psi\,|^2.
\end{split}
\end{align}
with $\Delta_x =\frac{\partial^2}{\partial x^2_1}+...+\frac{\partial^2}{\partial x^2_3}$ and $\eta = M\sigma^2$. Then the expectation of the pre-collision energy $E[\widetilde{\h}_1]$ = $E[\widetilde{\h}_2]$ and the forward drift $b^+$ and backward drift $b^-$ in \eqref{l:DRFT2aa}, \eqref{l:DRFT2ab} can be written as
\begin{align*}
b^\pm & = \frac{\eta}{M}\left(\IM \pm \gamma_m\RE \right)\frac{\nabla_x\psi}{\psi}.
\end{align*}
The three-dimensional main and incident energy equation \eqref{l:HAMILT1} becomes
\end{subequations}
\begin{align}
\label{l:HAMILT44}
\begin{split}
E\left[\widetilde{\h}_1\right] & = E\left[\widetilde{\h}_2\right] = \frac{M}{2}
\left(E\abs{\nabla_x S}^2 + E\abs{\nabla_x R}^2\right)
+\int\rho(x,t) V(x)dx + \frac{3\eta}{\overline{\tau}}
\\ & =
\frac{\eta^2}{2M}\int \abs{\nabla_x
\psi}^2 dx + \int \abs{\psi}^2 V(x)dx + \frac{3\eta}{\overline{\tau}}.
\end{split}
\end{align}
\begin{proof}
Details on the derivation of the Schr\"{o}dinger equation \eqref{l:SCHROD5} is in Appendix D where it was demonstrated that the average classical energy \eqref{l:HAMILT1} becomes time-invariant as long as the wave function of the probability density satisfies \eqref{l:SCHROD5}. The original proof by Nelson relied on differential equations whereas Appendix D minimizes the energy potential where the classical energy expression was created by the momentum equation in Theorem \ref{l:THEOREM3}. Carlen~\cite{ECARL1} demonstrated that the stochastic differential equation with drifts \eqref{l:SCHROD5} admits many strong and weak solutions for Rellich class potentials.

Taking \eqref{l:APP13a}, \eqref{l:APP13b} and substituting that into \eqref{l:HAMILT1} generates
\begin{align*}
E[\widetilde{\h}_1]=E[\widetilde{\h}_2]
=&\frac{M}{2}E\left[\abs{\nabla_{x} S}^2 +\frac{\eta^2}{4M^2}\abs{\frac{\nabla_{x}\rho}{\rho}}^2\right]
+\frac{3M\sigma^2}{2\overline{\tau}}
\\
=& \frac{M}{2}\left(E\abs{\nabla_x S}^2 + E\abs{\nabla_{x} R}^2\right)
+\frac{3\eta}{\overline{\tau}},
\end{align*}
where $\tau$ has an appropriate gamma distribution and $E\left[1/\tau\right]=2/\overline{\tau}$ and $\eta=M\sigma^2$ showing the average interaction time between collisions. This demonstrates the non-potential part of \eqref{l:HAMILT44}.
\end{proof}
\end{thm}

\section{Double Quantum Mechanics}
The previous section approximated the $w_1, w_2$ term with a non-random potential to show the relation to quantum mechanics but now assume that the $w_1, w_2$ terms are also undetermined. The easiest way to incorporate the incoming incident particle behaviour $w_2, w_1$ is to assume that the incident particle is represented by a diffusion process like the main particle and apply the same approach as in the last section. The two-particle Hamiltonian describing the classical elastic collision in Theorem \ref{l:THEOREM5} has four terms where the first two represent the main particle energy and the third and fourth terms represent the energy of the incident particle. In this case the four classical energy terms can be approximated by the  Schr\"{o}dinger again once for the main particle and once for the incident particle.

As the treatment of the incident particle will be the same as in the last section we shall only present the result of the procedure. Let $\eta_1 = M\sigma_1^2, \eta_2 = M\sigma_2^2$ and let
\begin{align*}
\psi_1 &=\psi_1(x,t)=
\exp\left[\frac{M\left(R_1(x,t) + iS_1(x,t)\right)}{\eta_1}\right],
\\
\psi_2 &=\psi_2(y,t)=
\exp\left[\frac{m\left(R_2(y,t) + iS_2(y,t)\right)}{\eta_2}\right],
\end{align*}
then without angular momentum or spin the energy Hamiltonian \eqref{l:DRFT4} in Theorem \ref{l:THEOREM5} will be time-independent under the Schr\"{o}dinger equations satisfying
\begin{align*}
i\eta_1\frac{\partial \psi_1(x,t)}{\partial t} & = -\frac{\eta_1^2}{2M}\Delta_{x}\psi_1(x,t)\psi_1(x,t),
\\
i\eta_2\frac{\partial \psi_2(y,t)}{\partial t}&=-\frac{\eta_2^2}{2m}\Delta_{y}\psi_2(y,t)\psi_2(y,t),
\end{align*}
where $\Delta_x =\frac{\partial^2}{\partial x^2_1}+...+\frac{\partial^2}{\partial x^2_3}$ , $\Delta_y =\frac{\partial^2}{\partial y^2_1}+...+\frac{\partial^2}{\partial y^2_3}$.
Here $\rho_1(x,t)=\abs{\psi_1(x,t)}^2$ represents the main particle position and $\rho_2(y,t)=\abs{\psi_2(y,t)}^2$ represents the incident particle. In this case
\begin{align}
\label{l:APP20}
\begin{split}
\frac{\partial E[\h_1]}{\partial t} & =\frac{\partial E[\h_2]}{\partial t} =
\\
& \frac{M}{2}\int\rho_1 (\nabla_{x} S_1)^T\nabla_{x}
\begin{pmatrix}
2S_{1,t}+\abs{\nabla_{x} S_1}^2 -\abs{\nabla_{x} R_1}^2 - \frac{\eta_1}{M}\Delta_{x} R_1
\end{pmatrix}
dx
\\
& +\frac{m}{2}\int\rho_2 (\nabla_{y} S_2)^T\nabla_{y}
\begin{pmatrix}
2S_{2,t}+\abs{\nabla_{y} S_2}^2 -\abs{\nabla_{y} R_2}^2 - \frac{\eta_2}{m}\Delta_{y} R_2
\end{pmatrix}
dy,
\end{split}
\end{align}
so that equation \eqref{l:HAMILT44} for the main and incident particle becomes
\begin{align}
\label{APP35}
\begin{split}
E\left[\h_1\right] = & E\left[\h_2\right] = \frac{M}{2}
\left(E\abs{\nabla_x S_1}^2 + E\abs{\nabla_x R_1}^2\right)
\\ & +
\frac{m}{2}
\left(E\abs{\nabla_y S_2}^2 + E\abs{\nabla_y R_2}^2\right)
+ \frac{3\eta_1}{\overline{\tau_1}}+\frac{3\eta_2}{\overline{\tau_2}}
\\ = &
\frac{\eta_1^2}{2M}\int \abs{\nabla_x
\psi_1}^2 dx +\frac{\eta_2^2}{2m}\int \abs{\nabla_y
\psi_2}^2 dy + \frac{3\eta_1}{\overline{\tau_1}}+\frac{3\eta_2}{\overline{\tau_2}},
\end{split}
\end{align}
showing how the energy changes due to the main and incident particle. Because the incident particle here is modelled separately, the combination of terms \eqref{APP35} must be constant in time. But notice the separate times $\tau_1, \tau_2$ and diffusions $\eta_1, \eta_2$ indicating the two stochastic processes.

It is possible to improve on the model \eqref{APP35} by inserting an input potential $V_1=V_1(x)$ and an output potential $V_2=V_2(y)$. Both particles are still independent and still use Schr\"odinger equation but now the potential added to the main particle provides energy while the potential for the incident particle also contributes energy. Because the processes are independent the energy contributions also separate and so equation \eqref{APP35} also holds for if there are independent potentials. In fact
\begin{align}
\label{APP50}
\begin{split}
E\left[\h_1\right] = & E\left[\h_2\right] = \frac{M}{2}
\left(E\abs{\nabla_x S_1}^2 + E\abs{\nabla_x R_1}^2\right)
+\int\rho_1(x,t) V_1(x)dx
\\ & +
\frac{m}{2}
\left(E\abs{\nabla_y S_2}^2 + E\abs{\nabla_y R_2}^2\right)
+\int\rho_2(y,t) V_2(y)dx + \frac{3\eta_1}{\overline{\tau_1}}+\frac{3\eta_2}{\overline{\tau_2}}
\\ = &
\frac{\eta_1^2}{2M}\int \abs{\nabla_x
\psi_1}^2 dx + \int \abs{\psi_1}^2 V_1(x)dx
\\ & +
\frac{\eta_2^2}{2m}\int \abs{\nabla_y
\psi_2}^2 dy + \int \abs{\psi_2}^2 V_2(y)dy + \frac{3\eta_1}{\overline{\tau_1}}+\frac{3\eta_2}{\overline{\tau_2}},
\end{split}
\end{align}
where the wave functions $\psi(x,t)$ and $\psi(y,t)$ satisfy the typical Schr\"odinger equations
\begin{align}
\label{APP51}
\begin{split}
i\eta_1\frac{\partial \psi_1(x,t)}{\partial t} & = -\frac{\eta_1^2}{2M}\Delta_{x}\psi_1(x,t)\psi_1(x,t) + V_1(x,t)\psi_1(x,t),
\\
i\eta_2\frac{\partial \psi_2(y,t)}{\partial t}&=-\frac{\eta_2^2}{2m}\Delta_{y}\psi_2(y,t)\psi_2(y,t) +V_2(y,t)\psi_2(y,t).
\end{split}
\end{align}

Obviously it is possible to create much more elaborate examples of particles that are correlated via momenta or energy levels. For instance, a correlation term between particles can be compensated for by energy dependencies for the incoming and outgoing particles modelled as free (zero-potential) particles. The diffusions $\eta_1, \eta_2$ do not have to be the same for the individual processes to make \eqref{APP50} to be time independent but it is easy to make \eqref{APP51} into Schr\"odinger equations by requiring that $\eta=\eta_1=\sigma_1^2M=\eta_2=\sigma_2^2m$ making all diffusions $\eta$ equal so that $\eta=\hbar$.

\section{Free Solutions and Correlation}
If the solution of Theorem \ref{l:THEOREM3} and Theorem \ref{l:THEOREM5} concerns classical mechanics it is clear that the collision process of the incident particle $m$ exchanges energy with the main particle $M$. This section addresses the many consequences of  by estimating the equilibrium solutions and determine the correlation between main and incident particles as a function of speed of the main particle. Interestingly, this correlation does not depend on the mass sizes assuming sufficiently small mass $m$.

First consider the case where there is no correlation between $v_1$ and $w_1$. Since
\begin{align*}
E[v_1^Tw_1]^2 \leq E\abs{v_1}^2E\abs{w_1}^2,
\end{align*}
the correlation can be defined as
\begin{align*}
\rho=\frac{E[v_1^Tw_1]}{\sqrt{E\abs{v_1}^2E\abs{w_1}^2}},
\end{align*}
and zero correlation means that $E[v_1^Tw_1]=0$.  The expectation of equation \eqref{APP15} for small mass ratio's $\gamma_m^2 \ll 1$ and $E[P(\phi)] \rightarrow I$ ($I$ unit matrix) becomes
\begin{align}
\begin{split}
\label{l:HAMILT8}
E[v_2] =
&
E[v_1] +
\gamma_m\sin{\theta}\left(E[w_1] - E[v_1]\right),
\\
=
&
E[v_1] + \frac{2\gamma_m^2}{(1+\gamma_m^2)}\left(E[w_1] -
E[v_1]\right)
\\
\approx
&
\left(1 - 2\gamma_m^2\right)E[v_1] + 2\gamma_m^2E[w_1],
\end{split}
\end{align}
so that
\begin{align}
\begin{split}
\label{l:HAMILT8}
E[v_2] - E[v_1] = 2\gamma_m^2( E[w_1] - E[v_1]),
\end{split}
\end{align}
which means the main particle $M$ experiences momentum acceleration or deceleration towards $E[w_1]$ as time increases.

Without correlation the energy calculation for $v_2$ becomes
\begin{align}
\label{APP19}
\begin{split}
E\abs{v_2}^2
= &
E\abs{v_1}^2 - 2\gamma_m\sin(\theta)E\abs{v_1}^2 +\gamma_m^2\sin^2(\theta) E\abs{v_1}^2
+\gamma_m^2\sin^2(\theta)E\abs{w_1}^2
\\
& +
2\gamma_m\sin(\theta)(1-\gamma_m\sin(\theta))E[v_1^Tw_1]
\\
= &
E\abs{v_1}^2 +\gamma_m\sin(\theta)\left(- 2 +\gamma_m\sin(\theta)\right) E\abs{v_1}^2
+\gamma_m^2\sin^2(\theta)E\abs{w_1}^2,
\end{split}
\end{align}
and then the denominator becomes
\begin{align*}
E\abs{v_2}^2 - E\abs{v_1}^2
& \approx
-4\gamma_m^2\left((1-\gamma_m^2)E\abs{v_1}^2
-\gamma_m^2E\abs{w_1}^2\right)
\\ & \approx
-4\gamma_m^2\left(E\abs{v_1}^2
-\gamma_m^2E\abs{w_1}^2\right).
\end{align*}
This equation and \eqref{l:HAMILT8} clearly show that
\begin{align*}
E\left[v_1\right]
& \approx
E\left[w_1\right],
\\
ME\abs{v_1}^2
& \approx
M\gamma_m^2c_w^2
=
mc_w^2,
\end{align*}
where $c_w^2 = E\abs{w_1}^2$. This means that for stationary processes the main particle assumes the incident particle speed and assumes the incident particle energy.

For the case that the main particle has a steady velocity a positive correlation becomes essential. Let $v_1,v_2$ and $w_1,w_2$ be three-dimensional then Theorem \ref{l:THEOREM3} shows that the after-collision speed for the main particle becomes
\begin{align}
\label{APP15}
v_2 =\left(I-\gamma_m\sin\left(\theta\right)P(\phi)\right)v_1+
\gamma_m\sin\left(\theta\right)P(\phi)w_1.
\end{align}
Assume that the energy is conserved in the incident particles during the particle collision then from the same equation \eqref{APP19} and a positive correlation  $E[v_1^Tw_1]$, it is clear that
\begin{align}
\label{APP26}
\begin{split}
0=E\abs{v_2}^2-E\abs{v_1}^2
= &- 2\gamma_m\sin(\theta)E\abs{v_1}^2 +\gamma_m^2\sin^2(\theta) E\abs{v_1}^2
+\gamma_m^2\sin^2(\theta)E\abs{w_1}^2
\\
& +
2\gamma_m\sin(\theta)(1-\gamma_m\sin(\theta))E[v_1^Tw_1],
\end{split}
\end{align}
hence if the incident vector always maintains the same energy there will be an effect on the main particle.

If the collision maintains the speed then $E\abs{v_2}^2 = E\abs{v_1}^2$ and $E[P(\phi)]\rightarrow I(\text{unit matrix})$ for the projection $P(\phi)$, then equation \eqref{APP26} reduces to
\begin{align*}
E[v_1^Tw_1]= &
\frac{
\left(2 -\gamma_m\sin(\theta)\right) E\abs{v_1}^2 - \gamma_m\sin(\theta)E\abs{w_1}^2
}
{
2\left(1-\gamma_m\sin(\theta)\right)
}
\approx
\frac{E\abs{v_1}^2-\gamma^2_m E\abs{w_1}^2}{1-\gamma_m^2},
\end{align*}
and the correlation becomes
\begin{align*}
\rho \approx \frac{\sqrt{E\abs{v_1}^2}}{\left(1-\gamma_m^2\right)\sqrt{E\abs{w_1}^2}}
-\frac{\gamma^2_m \sqrt{E\abs{w_1}^2}}{\left(1-\gamma_m^2\right)\sqrt{E\abs{v_1}^2}}.
\end{align*}

For small $\gamma_m^2 << 1$ this reduces to
\begin{align*}
\rho \approx \frac{\sqrt{E\abs{v_1}^2}}{\sqrt{E\abs{w_1}^2}} \approx
\frac{E\abs{v_1}}{E\abs{w_1}},
\end{align*}
assuming that $m \ll M$, $E\abs{v_1}^2 \approx E\abs{v_1}E\abs{v_1}$ and $E\abs{w_1}^2 \approx E\abs{w_1}E\abs{w_1}$. This last assumption removes any variance effects for the incident particle $m$ so for a very homogeneous medium the correlation will decrease further.

\section{The Minkowski transaction}
This section introduces classical elastic collisions for the main particle of mass $M$ and the incident particle $m$ where normally $m \ll M$. In this section a new method for finding the energy minimum is introduced using the Eigenvector representation of the collision matrix found in the momentum constraint \eqref{l:EQMOT9}. One Eigenvector for the collision matrix equals $1$ showing the amount of velocity maintained notwithstanding the collision. The other Eigenvector equals $-1$ showing the part of the motion that is explicitly changes sign upon collision.

First show the Eigenvectors for the collision matrix.
\begin{thm}
\label{l:THEOREM2}
Let the main and incident particles $v_1,v_2,w_1,w_2 \in \Re^3$ satisfy \eqref{l:ENERG1a}, \eqref{l:ENERG1b}, \eqref{l:ENERG1c}, \eqref{l:ENERG1d} and show classical momentum and energy conservation from the elastic collision in the form $\h_1(x)=\h_2(x)$ and $\mathcal{P}_1(x)=\mathcal{P}_2(x)$. Then equation \eqref{l:EQMOT9} can be solved as
\begin{align}
\label{l:EQMOT4}
\begin{split}
v_1 & =a-\gamma_m^2g,
\\
w_1 & =a+g,
\\
v_2 & =a+\gamma_m^2g^\bot,
\\
w_2 & =a-g^\bot,
\end{split}
\end{align}
where $g^\bot=g + \frac{1}{\gamma_m}\Phi$ and $\Phi$ is defined in Theorem \ref{l:THEOREM3} and $\abs{g}^2 = \abs{g^\bot}^2$. The vectors $(a,a)$ and $(-\gamma_m^2g,g)$ are Eigenvectors of the collision matrix with Eigenvalues $1,-1$. Here the vector $a$  becomes the average system velocity while $g$ and $g^\bot$ constitute the interactions between the main and incident particles.

\begin{proof}
The details of this Theorem can be found in Appendix E. Equation  $\eqref{l:EQMOT9}$ is linear and can be solved using Eigenvalues and Eigenvectors of the collision matrix. The energy conservation condition then specifies that $g$ and $g^\bot$ have identical absolute sizes.
\end{proof}
\end{thm}

This example shows that some form of the Minkowski surface condition by equations \eqref{l:EQMOT9} using Theorem \ref{l:THEOREM2}. It is very interesting to see that the Theorem shows a stochastic change in the Minkowski equation.
\begin{thm}
\label{l:THEOREM10}
The requirements on Theorem \ref{l:THEOREM2} specifies the mean motion of the mean collision mean $a$ and the interactions $g=w_1-v_1$, $g^\bot=g+\frac{1}{\gamma_m}\Phi$ by manipulating equation \eqref{l:EQMOT4} to derive
\begin{align}
\label{APP23}
\begin{split}
g & = \frac{w_1-v_1}{(1+\gamma_m^2)},
g^\bot = \frac{v_2-w_2}{(1+\gamma_m^2)},
\\
a & = \frac{Mv_2+mw_2}{M+m} = \frac{Mv_1+mw_1}{M+m},
\end{split}
\end{align}
where $g^\bot=g + \frac{1}{\gamma_m}\Phi$, $\abs{g}^2 = \abs{g^\bot}^2$ and where the function $\Phi$ is defined in Theorem \ref{l:THEOREM3}.
Then
\begin{align}
\label{APP21}
\abs{w_2-a}^2 - \abs{v_2 - a}^2 = \abs{w_1-a}^2 - \abs{v_1 - a}^2,
\end{align}
so that the motion of all particles is constrained vis-a-vis the average velocity $a$. If the main particle on average does not lose or gain energy, if the average velocity $a$ randomizes against the behaviour of $g+g^\bot$ or if the statistical correlation between $a$ and $b+b^\bot$ is zero then
\begin{align}
\label{APP22}
E\abs{w_2}^2 - E\abs{v_2}^2 = E\abs{w_1}^2 - E\abs{v_1}^2.
\end{align}
For small $m$ where $m<<M$ it is clear that equation \eqref{APP21} is very close to the Minkowski metric in equation \eqref{APP22}.
\begin{proof}
The proof is provided in Appendix F.
\end{proof}
\end{thm}
It is interesting to see that equation \eqref{APP21} holds true under all circumstances and that Minkowski's metric in equation \eqref{APP22} depends on a statistical statement suggesting that $w_1,w_2$ are independent of $(v_1,v_2) \approx (a,a)$ while $w_1, w_2$ are statistically independent from $v_1,v_2$. As is well known, Einstein solved the Special Relativistic case for gravity in equation \eqref{APP22} by varying the time to compensate assuming that the incident particle speed always has the same speed setting $\abs{w_1}^2=c_w^2\tau_1^2, \abs{w_2}^2=c_w^2\tau_2^2$. Using the pre - and post collision vector of the macroscopic object speeds $v_1$ and $v_2$ then with $c_w^2$ as the constant speed of light equation \eqref{APP22} becomes
\begin{align*}
c_w^2\tau_2^2 - \abs{v_2}^2\tau_2^2 = c_w^2\tau_1^2- \abs{v_1}^2\tau_1^2.
\end{align*}
This requirement determines the behaviour of the times $\tau_1$ and $\tau_2$ by setting
\begin{align}
\label{APP24}
\frac{\tau_2}{\tau_1} = \frac{\sqrt{1-\frac{\abs{v_1}^2}{c_w^2}}}{\sqrt{1-\frac{\abs{v_2}^2}{c_w^2}}},
\end{align}
and then for a particle that starts at zero $v_1=0$ the time for the post-collision particle becomes
\begin{align*}
\tau_2=\frac{\tau_1}{\sqrt{(1 - \frac{\abs{v_2}^2}{c_w^2})}},
\end{align*}
showing how the time changes for an object as the result of speed $v_2$. This Special Relativistic expression will give rise to
\begin{align*}
M_2=\frac{M_1}{\sqrt{(1 - \frac{\abs{v_2}^2}{c_w^2})}},
\end{align*}
and it becomes clear that
\begin{align*}
E=\overline{M} c_w^2,
\end{align*}
where $\overline{M} = M_2-M_1$ (see Wheeler~\cite{Wheeler1} for a proof).

\section{Conclusions}

This paper shows that the Edward Nelson’s stochastic mechanics approach to quantum mechanics
applies to all physics and therefore has applications in many other fields like collision problems, gas dynamics, cosmology and relativity. The quantum mechanical approximation in section 2 folded the effect of the colliding particles w2, w1 into a non-random potential. Then stochastic mechanics pre-scribed using Ito processes to obtain the pre-collision and post-collision velocities but especially for repeated collisions there are other stochastic processes that model the velocity and position simultaneously. Clearly the model at the moment considers only the momentum and energy transfer but for a classical collision it is important to model the angular momentum in the collision. The angular momentum may live in harmony with the motion momentum of the particle and change the correlation discovered in section 4.

This paper shows that the Edward Nelson’s stochastic mechanics approach to quantum mechanics applies to all classical physics and therefore has applications in classical physics, collision problems, in cosmology and relativity. The quantum mechanical approximation in section 2 used Ito processes but especially for repeated collisions other processes are possible. The derivation of the Minkowski equation in Section 5 depends on the classical double particle elastic collision modelled with stochastic processes that can ultimately be interpreted as quantum mechanical representations. In cosmology Mordehai Milgrom argued with many others Paredes~\cite{PAREDES1}, Cresson~\cite{CRESSON1}, ~\cite{CRESSON2}, Chavanis~\cite{CHAVANIS1}, Chamaraux~\cite{CHAMARAUX1} for a change Newton's Law or a change in relativity. However the presence of diffusion in physical motion  as suggested in this paper add changes to Newtonian laws as the diffusion of the star system tends to push out stars.

Equation \eqref{l:DRFT4} is a very curious equation as it specifies the energy exactly the way that Nelson suggested in 1966 but there are no expectations around this equation because there is no stochastic process implied. In Nelson's approach the pre-collision and post-collision velocities were considered parts of the Ito equation describing the forward and backward drift of the main particle $v_1,v_2$. However, in equation \eqref{l:DRFT4} the velocities $w_1, w_2$ are associated with the incident particle and are not generated from a diffusion equation. The Nelson requirement on the pre-collisions $v_1, w_1$ and the post-collision velocities $v_2,w_2$ express the energy preserved exactly in a classical world. Equation \eqref{l:DRFT4} does not require a statistical expectation and is clearly based on a classical momentum equation hence this energy conservation expression applies to all classical physics.

In addition, the proof of the energy constraint \eqref{l:DRFT4} in Appendix C using stochastic mechanics shows that the main energy of the two-particle system is embedded in the first two terms of \eqref{l:DRFT4} while the terms $w_1, w_2$ on the right of this equation represent the incident energy. The four terms two-particle classical energy in Theorem \ref{l:THEOREM5} are weighted by mass $M$ to insure that the main and incident particles express the system energy. However, the proof shows that neither of these double terms represent the main or incident energy of the particle as there are additional terms in \eqref{APP30} and \eqref{APP31}. This is not surprising in an elastic collision exchanging energy but it runs counter to the typical Nelson interpretation.

Clearly, quantum mechanics can be derived from expression \eqref{l:DRFT4} by assuming that the second process is represented as a non-random potential while assuming that the main particle has forward and backward drifts corresponding to the collision process. The remaining difference then becomes the process diffusion and the choice of $\eta = \hbar$ completes the analogy. In this case the forward and backward drift of the stochastic equations (Ito processes) are identified as the pre-collision and post-collision velocities of the collision showing that the expected energy is time-independent on average. There are no other processes part of the collision so the energy of the process must be a constant throughout the elastic collision. Using the positions of the main and incident particle with diffusion processes the average energy is invariant to time if the position of the main and incident particle both satisfy Schr\"{o}dinger's equation with the same diffusion variant $\eta=\hbar=M\sigma^2$.

As section 2 shows the classical elastic energy constraint in equation \eqref{l:DRFT4} is not unique. This changes the solution in terms of parameter sizes but not its observed proportions hence the required stochastic solutions can change in scale but not in nature. Equation \eqref{l:DRFT4} relates to a classical collision but that suggests that many collisions occur over a period of time and the question is how many collisions does the main particle experience and how does the incident particle behave upon repeated collisions? In this case the incident particle may find $w_2,w_1$ in different phases or it may be possible that $w_2,w_1$ and $v_2,v_1$ move in a correlated fashion. This is certainly possible in a world where the angular momentum is considered to be part of the energy equation. Given the result it is also possible that a different representation than Schr\"odinger's equation depending on a different statistical distribution that also allows a representation of velocity.

In section 3 it is assumed that the incident particle also satisfies an independent process with a random distribution like the main particle. The energy equation now takes an average over two distributions associated withe main and incident particle. Equation \eqref{APP35} shows that the energies are added linearly and have no interaction at all if the process for the main particle and the incident particle are independent. This simplifies the resulting energy distribution because it shows that for independent particles the energy of individual particles can be added directly to the system energy distribution. That means that the average velocity terms and osmotic averages (random components) can all be added together to derive the energy with only statistical independence. So the conserved collision energy can be expressed as two mean velocities, two expected values showing the diffusion process and two potentials as shown in equation \eqref{l:HAMILT1}. This energy becomes time-independent if the main particle satisfies Schr\"{o}dinger's equation with a variance $\eta=M\sigma^2$.

The conserved four (weighted) particle classical energy constraint implies correlations for stationary cases which is shown in section 4. The equations suggest that the standard energy and momentum functions imply that there is a standard correlation between $v_1$ and $w_1$ depending on the speed with which the main particle moves. If the main particle is large compared to the incident one it is clear that the speed of the main particle $v_1$ has a correlation $\rho$ to $w_1$ proportional to the speed of $\rho=\abs{v_1}/\abs{w_1}$. In other words, if the main particle $M$ can speed through the medium of incident particles at constant speed the correlation between incident particles $w_1, w_2$ and the main particle $v_1,v_2$ is proportional to $\abs{v_1}/\abs{w_1}$.

In principle, quantum mechanics can be derived from assuming that $\eta=\hbar$, but that does not define the quantal process due to the terms inside the time-dependent Schr\"{o}dinger equation. There are other objections to this approach as Wallstrom~\cite{WALLSTROM1} argued that additional conditions on the angular momentum do not follow from the stochastic derivation. Though this is a clear requirement on the theory, it is also evident that Nelson does not deal with angular momentum constructions. However, the angular momentum can be included relatively easily into equation \eqref{l:HAMILT1} or Theorem \ref{l:THEOREM5} as either separate angular momentum (squared) or spin terms. In addition a momentum equation must be derived that specifies the behaviour of the angular momentum and its effect on motion. With these two models the angular momenta and spin will be specified and the Wallstrom objection will be countered.

In section 5 the paper solves the momentum equations of the momentum stochastic differential using the collision matrix Eigenvectors and Eigenvalues and adding the remainders. The first Eigenvector with Eigenvalue 1 equals the average velocity field $a$ covering all input and output velocities and the second Eigenvector with Eigenvalue -1 provides the individual velocity changes in terms of velocity differences $(v_1-w_1)$ and $(v_2-w_2)$. Using the Eigenvalue / Eigenvector method the paper demonstrates that the energy constraint can be used to derive the Minkowski metric with an extra term depending on correlation between the average velocity and the incoming / outgoing particles. If the correlation approaches zero then this expression changes into the Minkowski metric otherwise a set of correlations will enter this equation. Notice that this  argument shows that the Nelson stochastic derivative approach is directly correlated with classical physics and clearly has implications on the side of relativity.


\medskip
\medskip
\subsection*{Acknowledgment}
The author Herschel Rabitz acknowledges support from the Army Research Office (W911NF-19-1-0382).

\newpage
\bibliographystyle{plain}
\bibliography{Xbib}

\newpage
\renewcommand{\theequation}{A.\arabic{equation}}
\setcounter{equation}{0}
\appendix
\section*{Appendix A}
\label{AppendixA}
\begin{center}
{\bf Theorem \ref{l:THEOREM3}: Momentum Elastic Collision - The Collision Matrix}
\end{center}
\medskip
\noindent This Appendix proves the momentum equation for the classical elastic collision between main and incident particle preserving momentum and energy in Theorem \ref{l:THEOREM3}. The method uses a projection classical elastic collision $P(\phi)$ assuming that both particle are perfect spheres and that the collision conserves momentum and energy but ignores angular momentum. The easiest approach is to separate the motion of three dimensional velocities $v_1, w_1, v_2, w_2$ along $\phi$ with $P(\phi)v_1$, $P(\phi)w_1$, $P(\phi)v_2$ and $P(\phi)w_2$ and present its perpendicular motion using $(I-P(\phi))v_1$, $(I-P(\phi))w_1$, $(I-P(\phi))v_2$ and $(I-P(\phi))w_2$. Then the collision is described by
\begin{align}
\label{l:ENERG4}
\begin{split}
v_1 & = (\phi\phi^T)v_1 + (v_1 - (\phi\phi^T)v_1)=P(\phi)v_1+\left(I-P(\phi)\right)v_1,
\\
w_1 & = (\phi\phi^T)w_1 + (w_1 - (\phi\phi^T)w_1)=P(\phi)w_1+\left(I-P(\phi)\right)w_1,
\\
v_2  & = (\phi\phi^T)v_2 + (v_2 - (\phi\phi^T)v_2)=P(\phi)v_2+\left(I-P(\phi)\right)v_2,
\\
w_2  & = (\phi\phi^T)w_2 + (w_2 - (\phi\phi^T)w_2)=P(\phi)w_2+\left(I-P(\phi)\right)w_2,
\end{split}
\end{align}
and the total momentum before and after the collision can be written as
\begin{align}
\label{l:ENERG1}
\begin{split}
\mathcal{M}_1(x) & = Mv_1 + mw_1
\\
& = MP(\phi)v_1+M(1-P(\phi))v_1 +
mP(\phi)w_1+m(1-P(\phi))w_1,
\\
\mathcal{M}_2(x) & = Mv_2 + mw_2
\\
& = MP(\phi)v_2+M(1-P(\phi))v_2 +
mP(\phi)w_2+m(1-P(\phi))w_2.
\end{split}
\end{align}

The pre-collision energy $\h_{1}(x)$ is given by $M\abs{v_1}^2$ (multiplied by a factor of 2) which upon being decomposed using $P(\phi)$ and $(I-P(\phi))$ becomes
\begin{align}
\label{l:ENERG15}
\begin{split}
2\h_{1}(x) = & M\abs{v_1}^2 + m\abs{w_1}^2
\\
= & M\abs{P(\phi)v_1+(I-P(\phi))v_1}^2 +  m\abs{P(\phi)w_1+(I-P(\phi))w_1}^2
\\
= & Mv_1^TP(\phi)P(\phi)v_1 +Mv_1^T(I-P(\phi))(I-P(\phi))v_1
\\
& + mw_1^TP(\phi)P(\phi)w_1 +mw_1^T(I-P(\phi))(I-P(\phi))w_1,
\\
 = & Mv_1^TP(\phi)v_1 +Mv_1^T(I-P(\phi))v_1 +
 mw_1^TP(\phi)w_1 +mw_1^T(I-P(\phi))w_1,
\end{split}
\end{align}
because $P(\phi)$ and $(I-P(\phi))$ are orthogonal projections so that $P(\phi)P(\phi)=P(\phi)$,
$(I-P(\phi))(I-P(\phi))=(I-P(\phi))$ and $P(\phi))(I-P(\phi)) = 0$. Similarly using the same projections the post-collision energy $\h_2(x)$ becomes
\begin{align}
\label{l:ENERG2}
\begin{split}
2\h_{2}(x) & = M\abs{v_2}^2+m\abs{w_2}^2
\\
& = Mv_2^TP(\phi)v_2+Mv_2^T(1-P(\phi))v_2 +
mw_2^TP(\phi)w_2+mw_2^T(1-P(\phi))w_2.
\end{split}
\end{align}

Since the motions $v_1-(\phi\phi^T)v_1 = (1-P(\phi))v_1$, $w_1-(\phi\phi^T)w_1=(1-P(\phi))w_1$
are orthogonal to the collision surface this part of the motion must remain unchanged during the collision. Hence,
\begin{align}
\label{l:ENERG17}
\begin{split}
(1-P(\phi))v_2 & = (1-P(\phi))v_1,
\\
(1-P(\phi))w_2 & = (1-P(\phi))w_1,
\end{split}
\end{align}
while $\phi^Tv_1, \phi^Tw_1$ collide with each other along $\phi$ yielding $\phi^Tv_2, \phi^Tw_2$. Putting \eqref{l:ENERG17} back into $\eqref{l:ENERG4}$ results in
\begin{align*}
v_1 & = P(\phi)v_1 + (I - P(\phi))v_1,
\\
w_1 & = P(\phi)w_1 + (I - P(\phi))w_1,
\\
v_2  & = P(\phi)v_2 + (I - P(\phi))v_1,
\\
w_2  & = P(\phi)w_2 + (I - P(\phi))w_1,
\end{align*}
implying that \eqref{l:ENERG1},  \eqref{l:ENERG15} and \eqref{l:ENERG2} change to
\begin{align*}
\mathcal{M}_1(x) & = Mv_1 + mw_1
\\
& = MP(\phi)v_1+ mP(\phi)w_1 + M(1-P(\phi))v_1
+m(1-P(\phi))w_1,
\\
\mathcal{M}_2(x) & = Mv_2 + mw_2
\\
& = MP(\phi)v_2 +
mP(\phi)w_2 + M(1-P(\phi))v_1 + m(1-P(\phi))w_1,
\\
2\h_{1}(x) & = M\abs{v_1}^2+m\abs{w_1}^2
\\
& = Mv_1^TP(\phi)v_1 + mw_1^TP(\phi)w_1 +
Mv_1^T(1-P(\phi))v_1 + mw_1^T(1-P(\phi))w_1,
\\
2\h_{2}(x) & = M\abs{v_2}^2 + m\abs{w_2}^2
\\
& = Mv_2^TP(\phi)v_2 +
mw_2^TP(\phi)w_2 + Mv_1^T(1-P(\phi))v_1 +mw_1^T(1-P(\phi))w_1,
\end{align*}
and with $\mathcal{M}_1(x) = \mathcal{M}_2(x)$ and $\h_1(x)=\h_2(x)$ this expression becomes
\begin{align*}
\begin{split}
MP(\phi)v_1+ mP(\phi)w_1 & =
MP(\phi)v_2 + mP(\phi)w_2,
\\
Mv_1^TP(\phi)v_1 + mw_1^TP(\phi)w_1
 & = Mv_2^TP(\phi)v_2 + mw_2^TP(\phi)w_2,
\end{split}
\end{align*}
or using $P(\phi)P(\phi)=P(\phi)=\phi^T\phi$ this becomes
\begin{align}
\label{l:APP3}
\begin{split}
M\phi^Tv_1+ m\phi^Tw_1 & =
M\phi^Tv_2 + m\phi^Tw_2,
\\
M(\phi^Tv_1)^2 + m(\phi^Tw_1)^2
 & = M(\phi^Tv_2)^2 + m(\phi^Tw_2)^2.
\end{split}
\end{align}

This may be interpreted as a balance of the initial momenta $M(\phi^Tv_1)$, $m(\phi^Tw_1)$ and initial energies $M(\phi^Tv_1)^2$, $m(\phi^Tw_1)^2$ against the final momenta $M(\phi^Tv_2)$, $m(\phi^Tw_2)$ and final energies $M(\phi^Tv_2)^2$, $m(\phi^Tw_2)^2$ while the orthogonal motion does not change with the collision. These are a one-dimensional collision and the following proposition shows the form for $\phi^Tv_2, \phi^Tv_2$.

\begin{prop}
\label{l:PROP1} The solution to the two equations in \eqref{l:APP3} equals
\begin{align}
\label{l:EQMOT1}
\begin{split}
\begin{pmatrix}
\phi^Tv_2 \\
\phi^Tw_2
\end{pmatrix}
&= \begin{pmatrix}
\cos(\theta) & \gamma_m\sin(\theta) \\
\frac{\sin(\theta)}{\gamma_m} & -\cos(\theta)
\end{pmatrix}
\begin{pmatrix}
\phi^Tv_1\\
\phi^Tw_1
\end{pmatrix}
= \Gamma_\theta
\begin{pmatrix}
\phi^Tv_1\\
\phi^Tw_1
\end{pmatrix},
\end{split}
\end{align}
defining $\Gamma_\theta$ with $\gamma_m^2=m/M, \sin\left(\theta\right) =2\gamma_m/\left(1+\gamma_m^2\right)$, $\cos\left(\theta\right)=\left(1-\gamma_m^2\right)/\left(1+\gamma_m^2\right)$ as defined in Theorem \ref{l:THEOREM3} above. Notice that this equation does not reflect the angular momentum basically assuming that this motion is independent of the particle momenta. From that point of view this equation is an approximation. This equation establishes equation \eqref{App14} above.

\begin{proof}
From the conservation of the momentum using equations \eqref{l:APP3} it is clear that
\begin{align*}
MP(\phi)v_2- MP(\phi)v_1=mP(\phi)w_1- mP(\phi)w_2,
\end{align*}
so that
\begin{align*}
M\phi^Tv_2- M\phi^Tv_1=m\phi^Tw_1- m\phi^Tw_2.
\end{align*}
Dividing this by the mass $M$ yields
\begin{align}
\label{l:REL13}
\phi^T(v_2 - v_1) = - \gamma_m^2\phi^T(w_2 - w_1),
\end{align}
so that
\begin{align*}
\phi^Tv_2 + \gamma_m^2\phi^Tw_2 =\phi^Tv_1 + \gamma_m^2\phi^Tw_1,
\end{align*}
using the definition of the mass ratio $\gamma_m$ in the Theorem.

From the energy distribution \eqref{l:APP3} it is clear that
\begin{align*}
Mv_2^TP(\phi)v_2 + mw_2^TP(\phi)w_2
=Mv_1^TP(\phi)v_1 + mw_1^TP(\phi)w_1,
\end{align*}
implying that
\begin{align*}
M(\phi^Tv_2)^2+m(\phi^Tw_2)^2
=M(\phi^Tv_1)^2+m(\phi^Tw_1)^2,
\end{align*}
which is equivalent to
\begin{align*}
M((\phi^Tv_2)^2-(\phi^Tv_1)^2)
=-m((\phi^Tw_2)^2-(\phi^Tw_1)^2),
\end{align*}
or
\begin{align*}
M(\phi^Tv_2-\phi^Tv_1)(\phi^Tv_2+\phi^Tv_1)
=-m(\phi^Tw_2-\phi^Tw_1)(\phi^Tw_2+\phi^Tw_1).
\end{align*}
Dividing by $M$ then yields
\begin{align*}
(\phi^Tv_2-\phi^Tv_1)(\phi^Tv_2+\phi^Tv_1)
=-\gamma_m^2(\phi^Tw_2-\phi^Tw_1)(\phi^Tw_2+\phi^Tw_1).
\end{align*}

However, the term $\gamma_m^2\phi^T(w_2-w_1)$ given by \eqref{l:REL13} above and can be substituted into the righthand side such that
\begin{align*}
(\phi^Tv_2-\phi^Tv_1)(\phi^Tv_2+\phi^Tv_1)
=(\phi^Tv_2-\phi^Tv_1)(\phi^Tw_2+\phi^Tw_1),
\end{align*}
and
\begin{align*}
(\phi^Tv_2-\phi^Tv_1)\left((\phi^Tv_2+\phi^Tv_1) - (\phi^Tw_2+\phi^Tw_1)\right)=0,
\end{align*}
establishing
\begin{align}
\label{l:REL14}
\phi^Tv_2+\phi^Tv_1 = \phi^Tw_2+\phi^Tw_1,
\end{align}
because this is a one-dimensional collision with $\phi^Tv_2 \neq \phi^Tv_1$. If $\phi^Tv_2 = \phi^Tv_1$ then $P(\phi)^Tv_2 = P(\phi)^Tv_1$ which means that $v_2=v_1$ using \eqref{l:ENERG17} and therefore $w_2=w_1$ so there is no interaction.

Combining equations \eqref{l:REL13} and \eqref{l:REL14} it is clear that
\begin{align*}
\phi^Tv_2 + \gamma_m^2\phi^Tw_2 & =\phi^Tv_1 + \gamma_m^2\phi^Tw_1,
\\
\phi^Tv_2-\phi^Tw_2& =-\phi^Tv_1  +\phi^Tw_1,
\end{align*}
or in matrix form
\begin{align*}
 \begin{pmatrix}
1 & \gamma_m^2 \\
1 & -1
\end{pmatrix}
\begin{pmatrix}
\phi^Tv_2
\\
\phi^Tw_2
\end{pmatrix}
=
 \begin{pmatrix}
1 & \gamma_m^2 \\
-1 & 1
\end{pmatrix}
\begin{pmatrix}
\phi^Tv_1
\\
\phi^Tw_1
\end{pmatrix}.
\end{align*}
Inverting the first matrix and multiplying the matrices shows that
\begin{align*}
\begin{pmatrix}
\phi^Tv_2
\\
\phi^Tw_2
\end{pmatrix}
=&
 \begin{pmatrix}
\frac{1-\gamma_m^2}{1+\gamma_m^2} & \frac{2\gamma_m^2}{1+\gamma_m^2} \\
\frac{2}{1+\gamma_m^2} & -\frac{1-\gamma_m^2}{1+\gamma_m^2}
\end{pmatrix}
\begin{pmatrix}
\phi^Tv_1
\\
\phi^Tw_1
\end{pmatrix}
=
\begin{pmatrix}
\cos(\theta) & \gamma_m\sin(\theta) \\
\frac{\sin(\theta)}{\gamma_m} & -\cos(\theta)
\end{pmatrix}
\begin{pmatrix}
\phi^Tv_1
\\
\phi^Tw_1
\end{pmatrix},
\end{align*}
using the definitions from the Theorem \ref{l:THEOREM3}. This proves Proposition \eqref{l:PROP1}.
\end{proof}
\end{prop}

It is relatively straightforward to extend this formulation to the 3-dimensional velocities $v_1$, $w_1$, $v_2$ and $w_2$ using the parts of the equation not involved in the collision. Multiply the equations \eqref{l:EQMOT1} with the vector $\phi$ to derive
\begin{align*}
\begin{pmatrix}
P(\phi)v_2 \\
P(\phi)w_2
\end{pmatrix}
&= \begin{pmatrix}
\cos(\theta) & \gamma_m\sin(\theta) \\
\frac{\sin(\theta)}{\gamma_m} & -\cos(\theta)
\end{pmatrix}
\begin{pmatrix}
P(\phi)v_1\\
P(\phi)w_1
\end{pmatrix}
=
\Gamma_\theta
\begin{pmatrix}
P(\phi)v_1\\
P(\phi)w_1
\end{pmatrix}.
\end{align*}
Then add $(1-P(\phi))v_2$, $(1-P(\phi))w_2$ on both sides and add the factors $(1-P(\phi))v_1-(1-P(\phi))v_1=0$ and $(1-P(\phi))w_1-(1-P(\phi))w_1=0$ under the vector yielding the following multi-dimensional version
\begin{align*}
\begin{pmatrix}
P(\phi)v_2 + (1-P(\phi))v_2\\
P(\phi)w_2 + (1-P(\phi))w_2
\end{pmatrix}
&=
\begin{pmatrix}
(1-P(\phi))v_2\\
(1-P(\phi))w_2
\end{pmatrix}
+
\\
& \begin{pmatrix}
\cos(\theta) & \gamma_m\sin(\theta) \\
\frac{\sin(\theta)}{\gamma_m} & -\cos(\theta)
\end{pmatrix}
\begin{pmatrix}
P(\phi)v_1+(1-P(\phi))v_1-(1-P(\phi))v_1\\
P(\phi)w_1+(1-P(\phi))w_1-(1-P(\phi))w_1
\end{pmatrix}.
\end{align*}
This can be simplified by $v_1=P(\phi)v_1+(1-P(\phi))v_1$, $w_1=P(\phi)w_1+(1-P(\phi))w_1$ and the matrix equality \eqref{l:ENERG17} to become
\begin{align}
\label{l:REL17}
\begin{split}
\begin{pmatrix}
v_2\\
w_2
\end{pmatrix}
& =
\begin{pmatrix}
(1-P(\phi))v_1\\
(1-P(\phi))w_1
\end{pmatrix}
+
\begin{pmatrix}
\cos(\theta) & \gamma_m\sin(\theta) \\
\frac{\sin(\theta)}{\gamma_m} & -\cos(\theta)
\end{pmatrix}
\begin{pmatrix}
v_1\\
w_1
\end{pmatrix}
-
\begin{pmatrix}
\cos(\theta) & \gamma_m\sin(\theta) \\
\frac{\sin(\theta)}{\gamma_m} & -\cos(\theta)
\end{pmatrix}
\begin{pmatrix}
(1-P(\phi))v_1\\
(1-P(\phi))w_1
\end{pmatrix}
\\
& =
\begin{pmatrix}
\cos(\theta) & \gamma_m\sin(\theta) \\
\frac{\sin(\theta)}{\gamma_m} & -\cos(\theta)
\end{pmatrix}
\begin{pmatrix}
v_1\\
w_1
\end{pmatrix}
+
\begin{pmatrix}
\gamma_m\sin(\theta) & -\gamma_m\sin(\theta)
\\
-\frac{\sin(\theta)}{\gamma_m} & \frac{\sin(\theta)}{\gamma_m}
\end{pmatrix}
\begin{pmatrix}
(1-P(\phi))v_1\\
(1-P(\phi))w_1
\end{pmatrix}
\\
& =
\begin{pmatrix}
\cos(\theta) & \gamma_m\sin(\theta) \\
\frac{\sin(\theta)}{\gamma_m} & -\cos(\theta)
\end{pmatrix}
\begin{pmatrix}
v_1\\
w_1
\end{pmatrix}
+
\begin{pmatrix}
\gamma_m\sin(\theta) (1-P(\phi))(v_1-w_1)\\
-\frac{\sin(\theta)}{\gamma_m} (1-P(\phi))(v_1-w_1)
\end{pmatrix}
\\
& =
\begin{pmatrix}
\cos(\theta) & \gamma_m\sin(\theta) \\
\frac{\sin(\theta)}{\gamma_m} & -\cos(\theta)
\end{pmatrix}
\begin{pmatrix}
v_1\\
w_1
\end{pmatrix}
+
\begin{pmatrix}
\gamma_m\Phi\\
-\frac{1}{\gamma_m}\Phi
\end{pmatrix},
\end{split}
\end{align}
where $\Phi = \sin(\theta) (1-P(\phi))(v_1-w_1)$. Notice that the terms
$(1-P(\phi))v_1$, $(1-P(\phi))w_1$ use $1-\cos(\theta)=\gamma_m\sin(\theta)$ as well as $1+\cos(\theta)=\sin(\theta)/\gamma_m$. From the third line of this statement we also have
\begin{align}
\label{l:DRFT20}
\begin{split}
\begin{pmatrix}
v_2\\
w_2
\end{pmatrix}
& =
\begin{pmatrix}
\cos(\theta) & \gamma_m\sin(\theta)
\\
\frac{\sin(\theta)}{\gamma_m} & -\cos(\theta)
\end{pmatrix}
\begin{pmatrix}
v_1\\
w_1
\end{pmatrix}
+
\begin{pmatrix}
\gamma_m\sin(\theta) (1-P(\phi))(v_1-w_1)
\\
-\frac{\sin(\theta)}{\gamma_m} (1-P(\phi))(v_1-w_1)
\end{pmatrix}
\\
& =
\begin{pmatrix}
\cos(\theta) + \gamma_m\sin(\theta) (1-P(\phi)) & \gamma_m\sin(\theta)-\gamma_m\sin(\theta) (1-P(\phi))
\\
\frac{\sin(\theta)}{\gamma_m} -\frac{\sin(\theta)}{\gamma_m} (1-P(\phi)) & -\cos(\theta)+\frac{\sin(\theta)}{\gamma_m} (1-P(\phi))
\end{pmatrix}
\begin{pmatrix}
v_1\\
w_1
\end{pmatrix}
\\
& =
\begin{pmatrix}
I - \gamma_m\sin(\theta)P(\phi) & \gamma_m\sin(\theta)P(\phi)
\\
\frac{\sin(\theta)}{\gamma_m}P(\phi) & I-\frac{\sin(\theta)}{\gamma_m}P(\phi)
\end{pmatrix}
\begin{pmatrix}
v_1\\
w_1
\end{pmatrix}.
\end{split}
\end{align}
Both equations \eqref{l:REL17} and \eqref{l:DRFT20} are represented in equation \eqref{l:EQMOT9}.

Notice that by definition $\Phi^T(v_2-v_1)=0$ as a result. This is easy to show using the definition for $v_2(t)$ in \eqref{l:REL17} and \eqref{l:DRFT20} and noticing that
\begin{align*}
v_2-v_1
& =(\cos(\theta)-1)v_1 + \gamma_m\sin(\theta)w_1+\gamma_m\Phi
\\
& =\gamma_m\sin(\theta)(w_1 -v_1)+\gamma_m\Phi.
\end{align*}
Now if $\Phi^T(v_2-v_1)=0$ then
\begin{align*}
0 =\gamma_m\sin(\theta)\Phi^T(w_1 -v_1)+\gamma_m\abs{\Phi}^2,
\end{align*}
with $\abs{\Phi}^2=\Phi^T\Phi$ and this follows because
\begin{align}
\label{l:DRFT24}
\begin{split}
\gamma_m \sin(\theta) & \Phi^T(w_1 -v_1)+\gamma_m\abs{\Phi}^2
\\
= & -\gamma_m\sin^2(\theta)(v_1 -w_1)^T(1-P(\phi))(v_1 -w_1)
\\
& +\gamma_m\sin^2(\theta)(v_1 -w_1)^T(1-P(\phi))(1-P(\phi))(v_1 -w_1)
\\
= & -\gamma_m\sin^2(\theta)(v_1 -w_1)^T(1-P(\phi))(v_1 -w_1)
\\
& +\gamma_m\sin^2(\theta)(v_1 -w_1)^T(1-P(\phi))(v_1 -w_1)
=0.
\end{split}
\end{align}
This establishes the proof for Theorem \ref{l:THEOREM3}.

\newpage
\renewcommand{\theequation}{B.\arabic{equation}}
\setcounter{equation}{0}
\section*{Appendix B}
\label{AppendixB}
\begin{center}
\noindent {\bf Theorem \ref{l:THEOREM4}: An alternative momentum equation}
\end{center}
\medskip
This Appendix shows how rewrite the momentum equation in a different format. Applying equation \eqref{l:EQMOT9} in Theorem \ref{l:THEOREM3} for velocity $v_2$ yields
\begin{align*}
v_2  & = (\phi\phi^T)v_2 + (v_1 - (\phi\phi^T)v_1)
\\
& = (v_1 - (\phi\phi^T)v_1) + \phi\left(\cos\left(\theta\right) (\phi^Tv_1)
+ \gamma_m \sin\left(\theta\right)(\phi^Tw_1)\right)
\\
& = v_1 + \phi\left((\cos\left(\theta\right)-1)(\phi^Tv_1)
+ \gamma_m \sin\left(\theta\right)(\phi^Tw_1)\right)
\\
& = v_1 + \phi\gamma_m \sin\left(\theta\right)
\left(-(\phi^Tv_1) + (\phi^Tw_1)\right)
\\
& = v_1 + \gamma_m \sin\left(\theta\right)
\left(P(\phi)w_1 - P(\phi)v_1\right)=v_1 +
\gamma_m\sin\left(\theta\right)P(\phi)(w_1-v_1),
\end{align*}
and
\begin{align*}
w_2  & = (\phi\phi^T)w_2 + (w_1 -
(\phi\phi^T)w_1)
\\
& = (w_1 - (\phi\phi^T)w_1) + \phi\left(\frac{\sin\left(\theta\right)}{\gamma_m} (\phi^Tv_1)
- \cos\left(\theta\right)(\phi^Tw_1)\right)
\\
& = w_1+ \phi\left(\frac{\sin\left(\theta\right)}{\gamma_m}
(\phi^Tv_1) - (\cos\left(\theta\right)+1)(\phi^Tw_1)\right)
\\
& = w_1+ \phi\frac{\sin\left(\theta\right)}{\gamma_m}\left(
(\phi^Tv_1) - (\phi^Tw_1)\right)
\\
& = w_1+ \frac{\sin\left(\theta\right)}{\gamma_m}\left(
P(\phi)v_1 - P(\phi)w_1\right)= w_1-
\frac{1}{\gamma_m}\sin\left(\theta\right)P(\phi)(w_1-v_1),
\end{align*}
combining $\phi\phi^Tv_1=P(\phi)v_1$,
$\phi\phi^Tw_1=P(\phi)w_1$. This proves equation \eqref{l:EQMOT1P}.

Also taking conditional expectations and substituting the approximations shows that
\begin{align*}
\begin{split}
\begin{pmatrix}
E[v_2 \vert x] \\
E[w_2\vert x]
\end{pmatrix}
\approx &
\begin{pmatrix}
E[v_1\vert x] \\
E[w_1\vert x]
\end{pmatrix}
+
\begin{pmatrix}
\gamma_m \sin\left(\theta\right) E[P(\phi)] (E[w_1\vert x]-E[v_1\vert x]) \\
- \frac{1}{\gamma_m}\sin\left(\theta\right)E[P(\phi)] (E[w_1\vert x]-E[v_1\vert x])
\end{pmatrix}
\\
\approx &
\begin{pmatrix}
E[v_1\vert x] \\
E[w_1\vert x]
\end{pmatrix}
+
\begin{pmatrix}
\gamma_m \sin\left(\theta\right) (E[w_1\vert x]-E[v_1\vert x]) \\
- \frac{1}{\gamma_m}\sin\left(\theta\right)(E[w_1\vert x]-E[v_1\vert x])
\end{pmatrix}
\approx
\begin{pmatrix}
E[v_1\vert x] \\
E[w_1\vert x]
\end{pmatrix},
\end{split}
\end{align*}
where $E[v_2\vert x], E[v_1\vert x]$, $E[w_2\vert x]$ and $E[w_1\vert x]$ are conditional distributions of velocities given location $x$ and $E[P(\phi)]\rightarrow I$. Clearly the $v_2,w_2$ will not change in distribution if the size of $v$ matches the size of $w_1$ otherwise particle $v$ will accelerate in speed to match the environment energy. If $v$ moves too quickly the environment will slow it down.

\newpage
\renewcommand{\theequation}{C.\arabic{equation}}
\setcounter{equation}{0}
\section*{Appendix C}
\label{AppendixC}
\begin{center}
\noindent {\bf Theorem \ref{l:THEOREM5}: The Classical Elastic Energy Constraint}
\end{center}
\medskip
This Appendix describes the energy of the classical elastic collision between the perfectly spherical main and incident particles. This equation can be verified by direct substitution however the following simple argument is more intuitive. First notice that from the definitions $\sin(\theta) = 2\gamma_m/(1+\gamma_m^2)$ and $\cos(\theta)=(1-\gamma_m^2)/(1+\gamma_m^2)$ it is straighforward that
\begin{align*}
\sin\left(\theta/2\right)=\frac{\gamma_m}{\sqrt{1+\gamma_m^2}},
\cos\left(\theta/2\right)=\frac{1}{\sqrt{1+\gamma_m^2}},
\end{align*}
so that
\begin{align*}
\frac{\sin\left(\theta/2\right)}{\cos\left(\theta/2\right)}=\gamma_m.
\end{align*}

To consider the sum of the squared 3-dimensional differences $v_2-v_1$ and $v_2+v_1$ for the main particle we first focus on the $v_2-v_1$ difference. Using \eqref{l:REL17} of Appendix A it is found that
\begin{align*}
\frac{v_2-v_1}{2}
&=\frac{1}{2} (\cos(\theta)-1)v_1+\frac{1}{2}\gamma_m\sin(\theta)w_1
+\frac{\gamma_m}{2}\Phi
\\
&=\sin\left(\theta/2\right)\big(-\sin\left(\theta/2\right)v_1
+\gamma_m\cos\left(\theta/2\right)w_1\big)
+\frac{\gamma_m}{2}\Phi
\\
&=\sin\left(\theta/2\right)A_p
+\frac{\gamma_m}{2}\Phi,
\end{align*}
with
\begin{align*}
A_p & = -\sin\left(\theta/2\right)v_1
+\gamma_m\cos\left(\theta/2\right)w_1
\\
& = \sin\left(\theta/2\right)\left(-v_1
+w_1\right),
\end{align*}
and as a result
\begin{align*}
\frac{v_2-v_1}{2\gamma_m}
&=\cos\left(\theta/2\right)A_p +\frac{1}{2}\Phi.
\end{align*}
But, we know that $\Phi^T(v_2-v_1)=0$ by Theorem
\ref{l:THEOREM3} so then
\begin{align*}
\sin\left(\theta/2\right)\Phi^TA_p
+\frac{\gamma_m}{2}\abs{\Phi}^2=0,
\end{align*}
or dividing by $\gamma_m$
\begin{align*}
\cos\left(\theta/2\right)\Phi^TA_p
+\frac{1}{2}\abs{\Phi}^2=0.
\end{align*}
Using this
\begin{align}
\label{l:APP8}
\begin{split}
\abs{\frac{v_2-v_1}{2}}^2
&=\gamma_m^2\left(\cos^2\left(\theta/2\right)\abs{A_p}^2
+\cos\left(\theta/2\right)\Phi^TA_p
+\frac{1}{4}\abs{\Phi}^2\right)
\\
&=\gamma_m^2\cos^2\left(\theta/2\right)\abs{A_p}^2 -\frac{\gamma_m^2}{4}\abs{\Phi}^2.
\end{split}
\end{align}

To get a similar expression for the average positive difference use \eqref{l:REL17} from Appendix A again to show that
\begin{align*}
\frac{v_2+v_1}{2}
& =\frac{1}{2} (\cos(\theta)+1)v_1+\frac{1}{2}\gamma_m\sin(\theta)w_1
+\frac{\gamma_m}{2}\Phi
\\
& = \cos\left(\theta/2\right)\big(\cos\left(\theta/2\right)v_1
+\gamma_m\sin\left(\theta/2\right)w_1\big) + \frac{\gamma_m}{2}\Phi
\\
& = \cos\left(\theta/2\right)B_p + \frac{\gamma_m}{2}\Phi,
\end{align*}
with
\begin{align*}
B_p & =\cos\left(\theta/2\right)v_1
+\gamma_m\sin\left(\theta/2\right)w_1
\\
& =\cos\left(\theta/2\right)\left(v_1
+\gamma_m^2w_1\right).
\end{align*}
Hence the difference squared becomes
\begin{align}
\label{l:APP9}
\abs{\frac{v_2+v_1}{2}}^2
& = \cos\left(\theta/2\right)^2\abs{B_p}^2 +
\gamma_m\cos\left(\theta/2\right)\Phi^TB_p+
\frac{\gamma_m^2}{4}\abs{\Phi}^2.
\end{align}

To get the sum of the squared differences add equation \eqref{l:APP8} and equation \eqref{l:APP9} to find
\begin{align}
\label{APP30}
\begin{split}
E_1 & =\abs{\frac{v_2+v_1}{2}}^2+\abs{\frac{v_2-v_1}{2}}^2
\\
& =\cos\left(\theta/2\right)^2\left(\gamma_m^2\abs{A_p}^2 + \abs{B_p}^2\right)
+\gamma_m\cos\left(\theta/2\right)\Phi^TB_p.
\end{split}
\end{align}
Now this exercise is repeated for $w_1, w_2$ to find the proper squared differences. First use \eqref{l:REL17} of Appendix A for the third time to show that
\begin{align*}
\gamma_m\frac{w_2-w_1}{2}
& =
\gamma_m\left(\frac{\sin(\theta)}{2\gamma_m}v_1
-\frac{1}{2}\left(\cos(\theta)+1\right)w_1
-\frac{1}{2\gamma_m}\Phi\right)
\\
& = \gamma_m\cos^2\left(\theta/2\right)\left(v_1
-w_1\right)
-\frac{1}{2}\Phi
\\
&=\cos\left(\theta/2\right)C_p
-\frac{1}{2}\Phi,
\end{align*}
with
\begin{align*}
C_p & = \gamma_m\cos\left(\theta/2\right)v_1
-\gamma_m\cos\left(\theta/2\right)w_1
\\
& = \sin\left(\theta/2\right)\left(v_1
-w_1\right) =-A_p.
\end{align*}
But we know that $\Phi^T(w_2-w_1)=0$ by momentum conservation $\Phi^T(v_2-v_1)=\gamma_m^2\Phi^T(w_1-w_2)=0$, such that
\begin{align*}
\cos\left(\theta/2\right)\Phi^TC_p
-\frac{1}{2}\abs{\Phi}^2=0.
\end{align*}
As a result
\begin{align}
\label{l:APP10}
\begin{split}
\gamma_m^2\abs{\frac{w_2-w_1}{2}}^2
& = \cos^2\left(\theta/2\right)\abs{C_p}^2
-\cos\left(\theta/2\right)\Phi^TC_p
+\frac{1}{4}\abs{\Phi}^2
\\
&=\cos^2\left(\theta/2\right)\abs{C_p}^2 -\frac{1}{4}\abs{\Phi}^2.
\end{split}
\end{align}

For the positive difference for the $w_2, w_1$ use equation \eqref{l:REL17} in Appendix A one more time from which it follows that
\begin{align*}
\frac{w_2+w_1}{2}
& =\frac{\sin(\theta)}{2\gamma_m}v_1-\frac{1}{2}\left(\cos(\theta)-1\right)w_1
-\frac{1}{2\gamma_m}\Phi
\\
& = \cos\left(\theta/2\right)\big(\cos\left(\theta/2\right)v_1
+\gamma_m\sin\left(\theta/2\right)w_1\big) - \frac{1}{2\gamma_m}\Phi
\\
& = \cos\left(\theta/2\right)D_p - \frac{1}{2\gamma_m}\Phi,
\end{align*}
with
\begin{align*}
D_p & =\cos\left(\theta/2\right)v_1
+\gamma_m\sin\left(\theta/2\right)w_1
\\
& =\cos\left(\theta/2\right)\left(v_1
+\gamma_m^2w_1\right) = B_p.
\end{align*}
Hence
\begin{align}
\label{l:APP11}
\abs{\frac{w_2+w_1}{2}}^2
& = \cos\left(\theta/2\right)^2\abs{D_p}^2 -
\frac{1}{\gamma_m}\cos\left(\theta/2\right)\Phi^TD_p+
\frac{1}{4\gamma_m^2}\abs{\Phi}^2.
\end{align}

Now adding the squares using \eqref{l:APP10} and
\eqref{l:APP11} results in
\begin{align}
\label{APP31}
\begin{split}
E_2 = & \gamma_m^2\left(\abs{\frac{w_2+w_1}{2}}^2+\abs{\frac{w_2-w_1}{2}}^2\right)
\\
& = \cos\left(\theta/2\right)^2\left(\abs{A_p}^2 + \gamma_m^2\abs{B_p}^2\right)
-\gamma_m\cos\left(\theta/2\right)\Phi^TB_p.
\end{split}
\end{align}
Finally
\begin{align*}
E_1+ E_2 & = \cos\left(\theta/2\right)^2\left(1+\gamma_m^2\right)\left(\abs{A_p}^2 + \abs{B_p}^2\right)= \frac{2\h_1}{M}
\\
& = \left(\abs{A_p}^2 + \abs{B_p}^2\right)= \frac{2\h_1}{M},
\end{align*}
because
\begin{align*}
\abs{A_p}^2 + \abs{B_p}^2
& = \abs{\sin\left(\theta/2\right)\left(-v_1
+w_1\right)}^2 + \abs{\cos\left(\theta/2\right)\left(v_1
+\gamma_m^2w_1\right)}^2
\\
& = \abs{v_1}^2 + \sin\left(\theta/2\right)^2\abs{w_1}^2 + \gamma_m^4\cos\left(\theta/2\right)^2\abs{w_1}^2
\\
& = \abs{v_1}^2 + \gamma_m^2\cos(\theta/2)^2\left(1 + \gamma_m^2\right)\abs{w_1}^2
\\
& = \abs{v_1}^2 + \gamma_m^2\abs{w_1}^2=\frac{2\h_1}{M}.
\end{align*}
This proves
Theorem \ref{l:THEOREM5}.

\newpage
\renewcommand{\theequation}{D.\arabic{equation}}
\setcounter{equation}{0}
\section*{Appendix D}
\label{AppendixD}
\begin{center}
\noindent {\bf Theorem \ref{l:HAMILT12}: The Schr\"{o}dinger Equation}
\end{center}
\medskip
This Appendix shows that a Schr\"{o}dinger equation for the position of the main particle \eqref{l:SCHROD5} solves the classical elastic energy solution \eqref{l:HAMILT1} following the proof as presented by Nelson in 1966 ~\cite{ENELSON3} and ~\cite{ENELSON2}. The proof is repeated here for the sake of completeness. Notice that equation \eqref{l:APP13a} suggests that $\rho\sim\exp{{\left[\frac{2 R}{\sigma^2}\right]}}$ where the amplitude $R$ is scaled to normalize the density and equation \eqref{l:APP13b} implies that $b^++b^-$ can be represented in a gradient of the phase $S=S(x,t)$. Then using \eqref{l:APP13a} and \eqref{l:APP13b} the expected energy \eqref{l:HAMILT1} can be represented as
\begin{align}
\label{l:HAMILT37}
\begin{split}
\frac{2}{M}E\left[\widetilde{\h}_1\right]
& = E
\begin{matrix}
\abs{\nabla_x S}^2
\end{matrix}
+
E\begin{matrix}
\abs{\nabla_x R}^2
\end{matrix}
+\frac{2}{M}\int\rho(x,t) V(x)dx,
\end{split}
\end{align}
ignoring the constant energy term. To become time-independent this classical conserved energy must have a zero derivative with respect to time hence creating time derivatives of $R(x,t)$, $S(x,t)$ and the partial derivative of $\rho=\rho(x,t)$ satisfying the continuity equation. The potential $V=V(x)$ is assumed to be time-dependent as this assumption makes the calculation easier.

Starting with time behaviour of the probability density it becomes clear from equations \eqref{l:APP13a} and \eqref{l:APP13b} that
\begin{align}
\label{l:HAMILT43}
\frac{\partial \rho}{\partial t}=\rho_t=-\nabla_x^T\left(\frac{b^++b^-}{2}\rho
\right)=-\nabla_x^T\left(\rho\nabla_x S \right),
\end{align}
from which follows
\begin{align}
\label{l:HAMILT41}
\begin{split}
\frac{\partial R}{\partial t} = &R_t=-\left(\frac{\sigma^2}{2}\Delta_x S + (\nabla_xS)^T\nabla_xR\right),
\\
\frac{\partial \nabla_x R}{\partial t}=&\nabla_x R_t=-\nabla_x\left(\frac{\sigma^2}{2}
\Delta_x S + (\nabla_xS)^T\nabla_xR\right),
\end{split}
\end{align}
where $\nabla_x^T=\left(\frac{\partial}{\partial x_1}, ...,\frac{\partial}{\partial x_3}\right)$ and $\Delta_x=\left(\frac{\partial^2}{\partial x_1^2} + \frac{\partial^2}{\partial x_2^2} + \frac{\partial^2}{\partial x_3^2}\right)$.

The time derivative of the energy (weighted by mass) \eqref{l:HAMILT37} yields
\begin{align*}
\frac{2}{M}\frac{\partial E\left[\widetilde{\h}_1\right]}{\partial t}
=&\frac{2}{M}\frac{\partial E\left[\widetilde{\h}_2\right]}{\partial t}
\\
=&\int\rho_t
\begin{pmatrix}
\abs{\nabla_x S}^2 + \abs{\nabla_x R}^2
\end{pmatrix}
dx  && \text{(a)}
\\
&+2\int\rho
\begin{pmatrix}
(\nabla_x S)^T\nabla_x S_t + (\nabla_x R)^T\nabla_x R_t
\end{pmatrix}
dx && \text{(b)}
\\
&+\frac{2}{M}\frac{d}{dt}\int\rho(x,t)V(x)dx=0, &&
\text{(c)}
\end{align*}
where $\rho_t$ equals $\rho_t=\frac{\partial \rho(x,t)}{\partial t}$ of the probability density and $S_t$ is the time-derivative of the phase and $R_t$ is the time-derivative of the amplitude. As there are no other processes in this collision the time derivative has to be zero. The three individual terms are examined in order starting with $(a)$.

Substituting $\rho=\exp\left[\frac{2R}{\sigma^2}\right]$ using equation \eqref{l:HAMILT41} the first term $(a)$ becomes
\begin{align*}
\text{(a)}
&=\int\frac{2}{\sigma^2}\rho R_t
\begin{pmatrix}
\abs{\nabla_x S}^2 + \abs{\nabla_x R}^2
\end{pmatrix}
dx
\\
&=\frac{2}{\sigma^2}\int\rho
\begin{pmatrix}
\left(-\frac{\sigma^2}{2} \Delta_x S -
(\nabla_x S)^T\nabla_x R\right)
\left(
\abs{\nabla_x S}^2 + \abs{\nabla_x R}^2
\right)
\end{pmatrix}
dx
\\
&=-\int \rho \Delta_x S
\begin{pmatrix}
\abs{\nabla_x S}^2 + \abs{\nabla_x R}^2
\end{pmatrix}
dx
-\frac{2}{\sigma^2}\int\rho (\nabla_x S)^T\nabla_x R
\begin{pmatrix}
\abs{\nabla_x S}^2 + \abs{\nabla_x R}^2
\end{pmatrix}
dx.
\end{align*}
So with one partial integral on the first term this reduces to
\begin{align*}
\text{(a)}
=&\int\rho (\nabla_x S)^T\nabla_x
\begin{pmatrix}
\abs{\nabla_x S}^2 + \abs{\nabla_x R}^2
\end{pmatrix}
dx
\\
&+\int (\nabla_x \rho)^T \nabla_x S
\begin{pmatrix}
\abs{\nabla_x S}^2 + \abs{\nabla_x R}^2
\end{pmatrix}
dx
\\
&-\frac{2}{\sigma^2}\int\rho (\nabla_x S)^T\nabla_x R
\begin{pmatrix}
\abs{\nabla_x S}^2 + \abs{\nabla_x R}^2
\end{pmatrix}
dx,
\end{align*}
or
\begin{align*}
\text{(a)} =&\int \rho (\nabla_x S)^T \nabla_x
\begin{pmatrix}
\abs{\nabla_x S}^2 + \abs{\nabla_x R}^2
\end{pmatrix}
dx
\\
&+\frac{2}{\sigma^2}\int \rho (\nabla_x R)^T \nabla_ xS
\begin{pmatrix}
\abs{\nabla_x S}^2 + \abs{\nabla_x R}^2
\end{pmatrix}
dx
\\
&-\frac{2}{\sigma^2}\int\rho
(\nabla_x R)^T \nabla_x S
\begin{pmatrix}
\abs{\nabla_x S}^2 + \abs{\nabla_x R}^2
\end{pmatrix}
dx,
\end{align*}
so that finally
\begin{align*}
\text{(a)}=&\int \rho (\nabla_x S)^T \nabla_x
\begin{pmatrix}
\abs{\nabla_x S}^2 + \abs{\nabla_x R}^2
\end{pmatrix}
dx.
\end{align*}

To reduce the $(b)$ term use equation \eqref{l:HAMILT41} again for $\nabla_x R$ and one partial derivative to yield
\begin{align*}
\text{(b)}= & 2\int\rho
\begin{pmatrix}
(\nabla_x S)^T \nabla_x S_{t}
- (\nabla_x R)^T\nabla_x\left(\frac{\sigma^2}{2}\Delta_x S + (\nabla_x S)^T\nabla_x R\right)
\end{pmatrix}
dx
\\
=&2\int\rho
\begin{pmatrix}
(\nabla_x S)^T \nabla_x S_{t} + \Delta_x R
\begin{pmatrix}
\frac{\sigma^2}{2} \Delta _x S + (\nabla_x S)^T\nabla_x R
\end{pmatrix}
\\
+\frac{2}{\sigma^2}\abs{\nabla_x R}^2
\begin{pmatrix}
\frac{\sigma^2}{2} \Delta _x S + (\nabla_x S)^T\nabla_x R
\end{pmatrix}
\end{pmatrix}
dx
\\
=&2\int\rho
\begin{pmatrix}
(\nabla_x S)^T \nabla_x S_{t}
+\frac{\sigma^2}{2}\Delta_x R \Delta_x S + \Delta_x R(\nabla_x S)^T\nabla_x R
\\
+ \abs{\nabla_x R}^2 \Delta_x S + \frac{2}{\sigma^2}\abs{\nabla_x R}^2 (\nabla_x S)^T\nabla_x R
\end{pmatrix}
dx,
\end{align*}
using one partial integral on the second term. Rearranging the terms this can be written as
\begin{align*}
\text{(b)}= &2\int\rho
\begin{pmatrix}
(\nabla_x S)^T\nabla_x S_t + (\nabla_x S)^T(\nabla_x R)\left(\Delta_x R +\frac{2}{\sigma^2}\abs{\nabla_x R}^2\right)
\end{pmatrix}
dx
\\
&+2\int\rho \Delta_x S
\begin{pmatrix}
\frac{\sigma^2}{2}\Delta_xR + \abs{\nabla_x R}^2
\end{pmatrix}
dx,
\end{align*}
and with one partial integral against the last term
this becomes
\begin{align*}
\text{(b)}= &2\int\rho
\begin{pmatrix}
(\nabla_x S)^T\nabla_x S_t + (\nabla_x S)^T(\nabla_x R)\left(\Delta_x R +\frac{2}{\sigma^2}\abs{\nabla_x R}^2\right)
\end{pmatrix}
dx
\\
&-2\int\rho
(\nabla_x S)^T\nabla_x
\begin{pmatrix}
\frac{\sigma^2}{2}\Delta_xR + \abs{\nabla_x R}^2
\end{pmatrix}
dx
\\
&-2\int\rho (\nabla_xS)^T\nabla_x R
\begin{pmatrix}
\Delta_xR +\frac{2}{\sigma^2} \abs{\nabla_x R}^2
\end{pmatrix}
dx.
\end{align*}
Clearly the second term and the fourth term cancel so $(b)$ reduces to
\begin{align*}
\text{(b)}= &2\int\rho
(\nabla_x S)^T\nabla_x S_t - (\nabla_xS)^T\nabla_x
\begin{pmatrix}
\frac{\sigma^2}{2}\Delta_xR + \abs{\nabla_x R}^2
\end{pmatrix}
dx.
\end{align*}

So then $(a)+ (b)$ becomes
\begin{align*}
\text{(a)}+\text{(b)}&
=\int \rho (\nabla_x S)^T \nabla_x
\begin{pmatrix}
\abs{\nabla_x S}^2 + \abs{\nabla_x R}^2
\end{pmatrix}
dx
\\
&+2\int\rho
\left((\nabla_x S)^T\nabla_x S_t - (\nabla_xS)^T\nabla_x
\begin{pmatrix}
\frac{\sigma^2}{2}\Delta_xR + \abs{\nabla_x R}^2
\end{pmatrix}
\right)
dx,
\end{align*}
which reduces to
\begin{align}
\label{l:DRFT17}
\begin{split}
\text{(a)}+\text{(b)}
= \int\rho (\nabla_x S)^T\nabla_x
\begin{pmatrix}
2S_t+\abs{\nabla_x S}^2 -\abs{\nabla_x R}^2 - \sigma^2\Delta_x R
\end{pmatrix}
dx,
\end{split}
\end{align}
and adding the second term $\text{(c)}$ to the results $(a) + (b)$ the time change of the potential becomes
\begin{align}
\label{l:DRFT131}
\begin{split}
\text{(a)}+\text{(b)}+\text{(c)}
=&\int \rho (\nabla_x S)^T\nabla_x
\begin{pmatrix}
2S_t+\abs{\nabla_x S}^2 -\abs{\nabla_x R}^2 - \sigma^2\Delta_x R
\end{pmatrix}
dx
\\
&+\frac{2}{M}\int\rho (\nabla_xS)^T
\nabla_x V(x)dx.
\end{split}
\end{align}
A sufficient condition for the Hamiltonian to become time-independent is that the integrand in \eqref{l:DRFT131} is to equal zero. In other words
\begin{align}
\label{l:SCHROD6}
S_{t}-\frac{1}{2}\abs{\nabla_x R}^2
+\frac{1}{2} \abs{\nabla_x S}^2 -\frac{\eta}{2M}\Delta_x R
+\frac{1}{M}V(x)=0,
\end{align}
defining $\eta = M\sigma^2$ and therefore implying that $\sigma^2 =\eta/M$.

The proof of Theorem \ref{l:HAMILT12} involves a straightforward verification of the equation \eqref{l:SCHROD5} from \eqref{l:SCHROD6}.  Expanding the derivatives on the wave function \eqref{l:DRFT23} shows that
\begin{align*}
\frac{\eta}{iM}\psi_t&=\psi\left( -i R_t+S_t\right),
\\
\frac{\eta}{2M}\nabla_x\psi&=\frac{1}{2}\psi\left( \nabla_xR+i\nabla_xS\right),
\\
\frac{\eta^2}{2M^2}\Delta_x\psi&=\frac{1}{2}\psi\abs{\nabla_xR+i\nabla_xS}^2 + \frac{\eta}{2M}\psi\left( \Delta_{x}R+i \Delta_{x}S\right),
\end{align*}
which combined becomes
\begin{align*}
\frac{\eta}{iM}\psi_t-\frac{\eta^2}{2M^2}\Delta_x\psi =
&-i\psi\left( R_t + (\nabla_xR)^T\nabla_xS+\frac{\eta}{2M}\Delta_xS\right)
\\
&+ \psi\left( S_t-\frac{1}{2}\abs{\nabla_x R}^2 +
\frac{1}{2}\abs{\nabla_x S}^2
-\frac{\eta}{2M}\Delta_xR\right)
\\
&=-\frac{V(x)}{M}\psi,
\end{align*}
because of \eqref{l:SCHROD6} and the definition of $R_t$ in \eqref{l:HAMILT41}. Applying the terms moving the second derivative to the other side reduces this equation to
\begin{align*}
\frac{\eta}{iM}\psi_t=\frac{\eta^2}{2M^2}\Delta_x\psi
-\frac{V(x)}{M}\psi,
\end{align*}
which after multiplication with the mass $M$ and the complex number $i^2=-1$ yields Schr\"{o}dinger's equation \eqref{l:SCHROD5} with $\eta = M\sigma^2$. Notice that
\begin{align*}
\rho(x,t)=\abs{\psi(x,t)}^2=e^{\frac{2M R}{\eta}}=e^{\frac{2R}{\sigma^2}},
\end{align*}
which is consistent with our original assumption.

To determine the current energy levels notice that equation \eqref{l:DRFT23}
implies that
\begin{align*}
\nabla_x \psi = \frac{M}{\eta}\left(\nabla_x S + i\nabla_x R\right)\psi,
\end{align*}
so that
\begin{align*}
\abs{\nabla_x \psi}^2 = \frac{M^2}{\eta^2}\left(\abs{\nabla_x S}^2 + \abs{\nabla_x R}^2\right)\abs{\psi}^2,
\end{align*}
and
\begin{align*}
\frac{\abs{\nabla_x \psi}^2}{\rho} = \frac{\abs{\nabla_x \psi}^2}{\abs{\psi}^2} = \frac{M^2}{\eta^2}\left(\abs{\nabla_x S}^2 + \abs{\nabla_x R}^2\right).
\end{align*}
Taking the expectation over this expression
\begin{align}
\label{APP27}
E\left[\frac{\abs{\nabla_x \psi}^2}{\abs{\psi}^2}\right] = \int\frac{\abs{\nabla_x \psi}^2}{\rho}\rho dx =
\int\abs{\nabla_x \psi}^2 dx.
\end{align}

To calculate the total energy \eqref{l:HAMILT37} invert equation \eqref{APP27} then
\begin{align*}
\int\abs{\nabla_x \psi}^2dx = & \frac{M^2}{\eta^2}\int\abs{\nabla_x S}^2\rho dx + \frac{M^2}{\eta^2}\int\abs{\nabla_x R}^2\rho dx
\\
= & \frac{M^2}{\eta^2}\left(E\left[\abs{\nabla_x S}^2\right] + E\left[\abs{\nabla_x R}^2\right]\right),
\end{align*}
so that
\begin{align*}
E\left[\widetilde{\h}_1\right] = E\left[\widetilde{\h}_2\right]
& = \frac{M}{2}
\left(E\abs{\nabla_x S}^2 + E\abs{\nabla_x R}^2\right) + \int\rho(x,t) V(x)dx
\\
&= \frac{M}{2}\frac{\eta^2}{M^2}
E\left[\abs{\frac{\nabla_x \psi}{\psi}}^2\right] + \int\rho(x,t) V(x)dx
\\
&= \frac{\eta^2}{2M}
\int\left[\abs{\nabla_x \psi}^2\right]dx + \int\rho(x,t) V(x)dx,
\end{align*}
which corresponds to \eqref{l:HAMILT37}.

\newpage
\renewcommand{\theequation}{E.\arabic{equation}}
\setcounter{equation}{0}
\section*{Appendix E}
\label{AppendixE}
\begin{center}
\noindent {\bf Theorem \ref{l:THEOREM2}: Eigenvector, Eigenvalue solution}
\end{center}
\medskip
This Appendix employs the Eigenvectors and Eigenvalues for the collision matrix to solve the classical energy constraint \eqref{l:DRFT4} and create a solution to the momentum equation. With the Eigenvectors $1, -1$ and Eigenvalues of the collision matrix in \eqref{l:EQMOT9} it is easy to see that
\begin{align}
\label{APP36}
\begin{split}
&\begin{pmatrix}
\cos(\theta) & \gamma_m\sin(\theta) \\
\frac{\sin(\theta)}{\gamma_m} & -\cos(\theta)
\end{pmatrix}
\begin{pmatrix}
a
\\a
\end{pmatrix}
=
\begin{pmatrix}
a\\a
\end{pmatrix},
\\ &
\begin{pmatrix}
\cos(\theta) & \gamma_m\sin(\theta) \\
\frac{\sin(\theta)}{\gamma_m} & -\cos(\theta)
\end{pmatrix}
\begin{pmatrix}
-\gamma_m^2g
\\
g
\end{pmatrix}
=
-\begin{pmatrix}
-\gamma_m^2g
\\
g
\end{pmatrix},
\end{split}
\end{align}
where the vectors $a$ and $g$ can be derived from the 3-dimensional vectors $v_1,v_2,w_1$ and $w_2$. In fact, to determine $g^\bot$ let
\begin{align}
\label{APP38}
\begin{pmatrix}
v_1
\\
w_1
\end{pmatrix}=
\begin{pmatrix}
a
\\
a
\end{pmatrix}
+
\begin{pmatrix}
-\gamma_m^2g
\\
g
\end{pmatrix},
\end{align}
then equation \eqref{l:EQMOT9} shows that $v_2$ and $w_2$ become
\begin{align}
\label{APP37}
\begin{split}
\begin{pmatrix}
v_2 \\
w_2
\end{pmatrix}
&=
\begin{pmatrix}
\cos(\theta) & \gamma_m\sin(\theta) \\
\frac{\sin(\theta)}{\gamma_m} & -\cos(\theta)
\end{pmatrix}
\begin{pmatrix}
v_1\\
w_1
\end{pmatrix}
+
\begin{pmatrix}
\gamma_m \Phi
\\
-\frac{1}{\gamma_m} \Phi
\end{pmatrix}
\\
&=
\begin{pmatrix}
a
\\
a
\end{pmatrix}
-
\begin{pmatrix}
-\gamma_m^2g
\\
g
\end{pmatrix}
+
\begin{pmatrix}
\gamma_m \Phi
\\
-\frac{1}{\gamma_m} \Phi
\end{pmatrix}
=
\begin{pmatrix}
a
\\
a
\end{pmatrix}
+
\begin{pmatrix}
\gamma_m^2(g+\frac{1}{\gamma_m}\Phi)
\\
-(g+\frac{1}{\gamma_m}\Phi)
\end{pmatrix}
=\begin{pmatrix}
a+\gamma_m^2g^\bot
\\
a-g^\bot
\end{pmatrix},
\end{split}
\end{align}
which is showm in equation \eqref{l:EQMOT4}.

The Eigenvectors $1,-1$ of the collision matrix in equation \eqref{APP36} and the solution \eqref{APP37} show that $M(1+\gamma_m^2)a = Mv_1+mw_1= Mv_2+mw_2$ while \eqref{APP38} suggests
$(1+\gamma_m^2)g=(w_1-v_1)$ and \eqref{APP37} shows that $(1+\gamma_m^2)g^{\bot}=-(w_2-v_2)$. Now consider the pre-collision energy $\h_1$ in the mean velocity $a$ and $g$ using equation \eqref{APP38} then
\begin{align*}
\frac{2}{M}\h_1
&=\abs{v_1}^2+\gamma_m^2\abs{w_1}^2=\abs{a-\gamma_m^2g}^2+\gamma_m^2\abs{a+g}^2
\\
&=(1+\gamma_m^2)\abs{a}^2 +\gamma_m^2(1+\gamma_m^2)\abs{g}^2.
\end{align*}
From equation \eqref{APP37} by the same procedure the post-collision energy can be determined
\begin{align*}
\frac{2}{M}\h_2
&=\abs{v_2}^2+\gamma_m^2\abs{w_2}^2=\abs{a+\gamma_m^2g^\bot}^2+\gamma_m^2\abs{a-g^\bot}^2
\\
&=(1+\gamma_m^2)\abs{a}^2 +\gamma_m^2(1+\gamma_m^2)\abs{g^\bot}^2,
\end{align*}
and since this is an elastic collision
\begin{align*}
\frac{2}{M}\h_1
&=(1+\gamma_m^2)\abs{a}^2 +\gamma_m^2(1+\gamma_m^2)\abs{g}^2
\\
&= (1+\gamma_m^2)\abs{a}^2 +\gamma_m^2(1+\gamma_m^2)\abs{g^\bot}^2=\frac{2}{M}\h_2,
\end{align*}
so that clearly
\begin{align*}
\gamma_m^2(1+\gamma_m^2)\abs{g}^2 = \gamma_m^2(1+\gamma_m^2)\abs{g^\bot}^2,
\end{align*}
and $\abs{g}^2=\abs{g^\bot}^2$. This proves Theorem \ref{l:THEOREM2}.

\newpage
\renewcommand{\theequation}{F.\arabic{equation}}
\setcounter{equation}{0}
\section*{Appendix F}
\label{AppendixF}
\begin{center}
\noindent {\bf Theorem \ref{l:THEOREM10}: The Minkowski Metric}
\end{center}
\medskip
Notice from equations \eqref{l:EQMOT4} it follows that
\begin{align}
\label{l:EQMOT5}
\begin{split}
v_1-w_1 &= (1+\gamma_m^2)g,
\\
v_1+\gamma_m^2w_1 &= (1+\gamma_m^2)a,
\end{split}
\end{align}
and similarly that
\begin{align}
\label{l:EQMOT6}
\begin{split}
v_2-w_2 &= -(1+\gamma_m^2)g^\bot,
\\
v_2+\gamma_m^2w_2 &= (1+\gamma_m^2)a,
\end{split}
\end{align}
and since $\abs{g}^2 = \abs{g^\bot}^2$ as a result
\begin{align*}
\begin{split}
\begin{split}
\abs{v_1-w_1}^2 &= \abs{v_2-w_2}^2,
\\
v_2+\gamma_m^2w_2 &= v_1+\gamma_m^2w_1,
\end{split}
\end{split}
\end{align*}
or
\begin{align*}
\begin{split}
\begin{split}
\abs{v_1-w_1}^2 &= \abs{v_2-w_2}^2,
\\
a = \frac{Mv_2+mw_2}{M+m} &= \frac{Mv_1+mw_1}{M+m}.
\end{split}
\end{split}
\end{align*}
The second equation shows the conservation of the mean momentum $a$ through the collision while the first equation shows that the noise terms $g, g^\bot$ are conserved through the term
$(1+\gamma_m^2)^2\abs{g}^2=\abs{w_1-v_1}^2=\abs{w_2-v_2}^2=(1+\gamma_m^2)^2\abs{g^\bot}^2$.

From this it is clear that
\begin{align}
\label{APP39}
\begin{split}
g & = \frac{v_1-w_1}{(1+\gamma_m^2)}, g^\bot = \frac{w_2-v_2}{(1+\gamma_m^2)},
\\
a & = \frac{Mv_2+mw_2}{M+m} = \frac{Mv_1+mw_1}{M+m},
\end{split}
\end{align}
where $g^\bot=g + \frac{1}{\gamma_m}\Phi$, $\abs{g}^2 = \abs{g^\bot}^2$ and where the function $\Phi$ is defined in Theorem \ref{l:THEOREM3}. This shows equation \eqref{APP23}.

To demonstrate equation \eqref{l:THEOREM2} use the Eigenvectors $1,-1$ of the collision matrix again and use equation \eqref{l:EQMOT4} to derive
\begin{align*}
\abs{w_2-a}^2 - \abs{v_2-a}^2
= \abs{-g^\bot}^2 - \abs{\gamma_m^2g^\bot}^2
= (1-\gamma_m^4)\abs{g^\bot}^2,
\end{align*}
and
\begin{align*}
\abs{w_1-a}^2 - \abs{v_1-a}^2
= \abs{g}^2-\abs{-\gamma_m^2g}^2
= (1-\gamma_m^4)\abs{g}^2.
\end{align*}
Then
\begin{align*}
\abs{w_2-a}^2 - \abs{v_2-a}^2 -\left(
\abs{w_1-a}^2 - \abs{v_1-a}^2\right)
=
(1-\gamma_m^4)(\abs{g^\bot}^2-\abs{g}^2)
=0,
\end{align*}
because $\abs{g}^2=\abs{g^\bot}^2$.  So
\begin{align*}
\abs{w_2-a}^2 - \abs{v_2 - a}^2 = \abs{w_1-a}^2 - \abs{v_1 - a}^2,
\end{align*}
which proves equation \eqref{APP21}.

To show the Minkowski metric in equation \eqref{APP22} notice the motion of all particles is constrained vis-a-vis the average velocity $a$. However, if $M >> m$ then the individual collision will not make any reason for the motion to change so $a$ is uncorrelated from $v_1,v_2$. To remove the dependence on the vector $a$ remove the vector from \eqref{APP21} then it is clear that
\begin{align*}
& \abs{w_2-a}^2 - \abs{v_2 - a}^2 - \left( \abs{w_1-a}^2 - \abs{v_1 - a}^2\right)
\\
& = \abs{w_2}^2 - \abs{v_2}^2 - \left(\abs{w_1}^2 - \abs{v_1}^2\right) + 2a^T\left(v_2-w_2 -(v_1-w_1)\right)
\\
& = \abs{w_2}^2 - \abs{v_2}^2 - \left(\abs{w_1}^2 - \abs{v_1}^2\right) -2(1+\gamma_m^2)a^T\left(g+g^\bot\right),
\end{align*}
so that
\begin{align*}
0 & = E\abs{w_2-a}^2 - E\abs{v_2 - a}^2 - \left( E\abs{w_1-a}^2 - E\abs{v_1 - a}^2\right)
\\
& = E\abs{w_2}^2 - E\abs{v_2}^2 - \left(E\abs{w_1}^2 - E\abs{v_1}^2\right) -2(1+\gamma_m^2)E\left[a^T\left(g+g^\bot\right)\right].
\end{align*}

However, the correlation between $a$ and $g+g^\bot$ is equal to
\begin{align*}
\rho=\frac{E\left[a^T(g+g^\bot)\right]}{\sqrt{E\abs{a}^2
E\abs{g+g^\bot}^2}},
\end{align*}
so that
\begin{align*}
E\abs{w_2}^2 - E\abs{v_2}^2 - \left(E\abs{w_1}^2 - E\abs{v_1}^2\right) =2(1+\gamma_m^2)\rho\sqrt{E\abs{a}^2
E\abs{g+g^\bot}^2},
\end{align*}
hence
\begin{align}
\label{APP40}
E\abs{w_2}^2 - E\abs{v_2}^2 = E\abs{w_1}^2 - E\abs{v_1}^2,
\end{align}
if $\rho$ equals zero. Notice that for large mass where $m\ll M$ it is clear that $a \approx v_1 \approx v_2$ and $b+b^\bot \approx w_1-w_2$ so then $E\left[a^T(g+g^\bot)\right] \approx E[v(w_2-w_1)] \approx vE[(w_2-w_1)] \approx 0$. So a large mass $M$ moving through space at a constant speed not loosing or gaining energy would experience $E\left[a^T(g+g^\bot)\right]=0$ and the correlation becomes zero.

However, there is another reason for \eqref{APP40} to hold. Write from equation \eqref{l:EQMOT4} the main particle velocity differences
\begin{align*}
\abs{v_1}^2 & = \abs{a}^2-\gamma_m^2a^Tg+\abs{g}^2,
\\
\abs{v_2}^2 & = \abs{a}^2+\gamma_m^2a^Tg^\bot+\abs{g^\bot}^2,
\end{align*}
then the amount of energy gained or lost equals
\begin{align*}
2\Delta E = M\left(\abs{v_2}^2-\abs{v_1}^2\right) = \gamma_m^2a^Tg^{\bot}+\abs{g}^2+\gamma_m^2a^Tg-\abs{g^\bot}^2
=\gamma_m^2a^T\left(g^{\bot}+g\right).
\end{align*}
This means that $\gamma_m^2a^T\left(g^{\bot}+g\right)$ shows the amount of energy leaving or arriving in the system so for a static system with $\Delta E=0$ the correlation is zero and the Minkowski metric in \eqref{APP40} holds.

Finally, it is interesting to see that for small $m$ where $m<<M,\gamma_m^2\rightarrow 0$ it is clear that $\Delta E$ is very close to zero. A large mass with mass $M$ moving through space meets almost exclusively by photons with mass $m=0$ and $\gamma_m=0$ so indeed $\Delta E\rightarrow0$ proving the Minkowski metric \eqref{APP40} again. This proves the remaining statements in Theorem \ref{APP23}.

\newpage
\end{document}